\documentclass[a4paper]{amsart}

\usepackage[foot]{amsaddr}

\usepackage{amsmath}
\usepackage{amssymb}

\usepackage{hyperref}

\usepackage{stmaryrd}
\usepackage{enumerate}
\usepackage[shortlabels]{enumitem}
\usepackage{microtype}

\newtheorem{lemma}{Lemma}
\newtheorem{theorem}[lemma]{Theorem}
\newtheorem{proposition}[lemma]{Proposition}
\newtheorem{corollary}[lemma]{Corollary}
\newtheorem{example}{Example}

\newcommand{\valX}[1]{\llbracket #1 \rrbracket}
\newcommand{\valXG}[2]{\valX{#2}_{#1}}
\newcommand{\rhs}{\rho}
\newcommand{\ftrees}{\text{hor}}
\newcommand{\fspine}{\text{spine}}
\newcommand{\fchop}{\triangle}
\newcommand{\rootof}[1]{\alpha_{#1}}
\newcommand{\nf}[1]{\text{nf}_{#1}}

\newcommand{\concv}{\mathop{\boxbar}}
\newcommand{\conch}{\mathop{\boxdot}}
\newcommand{\mergh}{\mathop{\varodot}}
\newcommand{\mergv}{\mathop{\varobar}}
\newcommand{\rank}{\text{rank}}

\newcommand{\auf}{(}
\newcommand{\zu}{)}
\newcommand{\lan}{\langle}
\newcommand{\ran}{\rangle}

\newcommand{\spineslp}[1]{#1^{\concv}}

\newcommand{\ribslp}[1]{#1^{\conch}}

\newcommand{\gspine}[2]{\fspine_{#1}(#2)}

\newcommand{\Up}{\mathcal{U}}
\newcommand{\Low}{\mathcal{L}}
\newcommand{\Hor}{\mathcal{H}}
\newcommand{\Ver}{\mathcal{V}}

\newcommand{\kl}[1]{{#1}_\ell}
\newcommand{\km}[1]{{#1}_m}
\newcommand{\kr}[1]{{#1}_r}
\newcommand{\ks}[1]{\sigma_{#1}}

\newcommand{\Kl}[1]{{#1}_\ell}
\newcommand{\Km}[1]{{#1}_m}
\newcommand{\Kr}[1]{{#1}_r}
\newcommand{\Ks}[1]{\sigma_{\!#1}}

\newcommand{\ksort}[1]{\text{sort}^{#1}}
\newcommand{\kchop}[1]{{#1}_\fchop}
\newcommand{\kcont}[1]{{#1}^{\pi}}
\newcommand{\ksibl}[1]{R_{#1}}

\newcommand{\eps}{\varepsilon}
\newcommand{\bx}{\mathrm{\triangledown_x}}
\newcommand{\addroot}[2]{#1 ( #2 )}

\newcommand{\llex}{\text{llex}}

\newcommand{\fcanon}{\nf{\mathcal{C}}}
\newcommand{\sleq}{\subseteq}
\newcommand{\nfc}{\fcanon}
\newcommand{\Comm}{\mathcal{C}}
\newcommand{\Down}[1]{M_{#1}}
\newcommand{\args}[1]{\text{args}(#1)}
\newcommand{\Mean}[2]{\valXG{#2}{#1}}

\newcommand{\nfa}{\nf{\mathcal{A}}}
\newcommand{\pull}{\phi}

\begin{document}

\title{Grammar-based Compression of Unranked Trees}


\author{Adri{\`a}~Gasc{\'o}n}
\email{agascon@inf.ed.ac.uk}

\author{Markus~Lohrey}
\email{lohrey@eti.uni-siegen.de}

\author{Sebastian~Maneth}
\email{sebastian.maneth@gmail.com}

\author{Carl~Philipp~Reh}
\email{reh@eti.uni-siegen.de}

\author{Kurt~Sieber}
\email{sieber@informatik.uni-siegen.de}

\address[Adri{\`a}~Gasc{\'o}n]{Warwick University and Alan Turing Institute, UK}
\address[Markus~Lohrey, Carl~Philipp~Reh, Kurt~Sieber]{Universit{\"a}t Siegen, Germany}
\address[Sebastian~Maneth]{Universit{\"a}t Bremen, Germany}

\begin{abstract}
  We introduce forest straight-line programs (FSLPs) as a compressed representation
  of unranked ordered node-labelled trees. FSLPs are based on the operations of forest algebra
  and generalize tree straight-line programs. We compare the succinctness of FSLPs with two other
  compression schemes for unranked trees: top dags and tree straight-line programs of first-child/next sibling
  encodings. Efficient translations between these formalisms are provided.
  Finally, we show that equality of
  unranked trees in the setting where certain symbols are associative
  or commutative can be tested in polynomial time. This generalizes previous results
  for testing isomorphism of compressed unordered ranked  trees.
\end{abstract}

\maketitle

\section{Introduction}

Generally speaking, grammar-based compression represents an object succinctly by means of a small context-free grammar.
In many grammar-based compression formalisms such a grammar can be exponentially smaller than the object.
Henceforth, there is a great interest in problems that can be solved in polynomial
time on the grammar, while requiring at least linear time on the original uncompressed object.
One of the most well-known and fundamental such problems is testing
equality of the strings
produced by two context-free string grammars, each producing exactly one string  (such grammars
are also known as straight-line programs --- in this paper we use the term
string straight-line program, SSLP for short).
Polynomial time solutions to this problem were discovered, in different contexts by different groups of people,
see the survey \cite{lohrey_survey} for references.

Grammar-based compression has been generalized from strings to ordered ranked node-labelled trees,
by means of linear context-free tree grammars generating exactly one tree~\cite{DBLP:journals/is/BusattoLM08}.
Such grammars are also known as tree straight-line programs, TSLPs for short, see~\cite{Lohrey15dlt} for a survey.
Equality of the trees produced by two TSLPs
can also be checked in polynomial time: one constructs SSLPs
for the pre-order traversals of the trees, and then applies the above mentioned result for SSLPs, see~\cite{DBLP:journals/is/BusattoLM08}.
The tree case becomes more complex when \emph{unordered} ranked trees are considered.
Such trees can be represented using TSLPs, by simply ignoring the
order of children in the produced tree. Checking isomorphism of unordered ranked trees
generated by TSLPs was recently shown to be solvable in polynomial
time~\cite{DBLP:conf/icalp/LohreyMP15}.
The solution transforms the TSLPs so that they generate canonical
representations of the original trees
and then checks equality of these canonical forms.

\newcounter{horiverti}
\renewcommand{\thehoriverti}{\roman{horiverti}}
\newcommand{\nextstephoriverti}[1]{\refstepcounter{horiverti}\thehoriverti\label{#1}}
The aforementioned result for ranked trees cannot be applied
to {\em unranked}
trees (where the number of children of a node is not bounded),
which arise for instance in XML document trees.
This is unfortunate, because ({\it i}\/)~grammar-based compression is particularly
effective for XML document trees (see~\cite{DBLP:journals/is/LohreyMM13}), and
({\it ii}\/)~XML document trees can often be considered unordered (one speaks
of ``data-centric XML'', see e.g.~\cite{DBLP:journals/mst/AbiteboulBV15,DBLP:conf/lata/BoiretHNT15,DBLP:journals/mst/BonevaCS15,DBLP:journals/dke/SundaramM12,DBLP:journals/tcyb/ZhangDW15}),
allowing even stronger grammar-based compressions~\cite{DBLP:conf/icdt/LohreyMR17}.

In this paper we introduce a generalization of TSLPs and SSLPs that allows to produce ordered unranked node-labelled trees
and forests (i.e., ordered sequences of trees) that we call {\em forest straight-line programs}, FSLPs
for short.  In contrast to TSLPs, FSLPs can compress very wide and flat trees. For instance,
the tree $f\auf a,a,\ldots, a\zu$ with $n$ many $a$'s is not compressible with
TSLPs but can be produced by an FSLP of size $O(\log n)$. FSLPs are based on the operations
of horizontal and vertical forest composition from forest
algebras~\cite{DBLP:conf/birthday/BojanczykW08}.
The main contributions of this paper are the following:

\subsection{\bf Comparison with other formalisms.} We compare the succinctness of FSLPs with
two other grammar-based formalisms for compressing unranked node-labelled ordered trees:
TSLPs for `first-child/next-sibling'' (fcns) encodings and top dags. The fcns-encoding is the standard
way of transforming an unranked tree into a binary tree. Then the resulting binary tree can be succinctly
represented by a TSLP.  This approach was used to apply the TreeRePair-compressor from~\cite{DBLP:journals/is/LohreyMM13}
to unranked trees. We prove that FSLPs and TSLPs for fcns-encodings are equally succinct
up to constant multiplicative factors and that one can change between both representations in linear time
(Propositions~\ref{lemma-fcns-transform} and \ref{lemma:fcns-transform-rev}).

Top dags are another formalism for compressing unranked trees~\cite{BilleGLW15}. Top dags use
horizontal and vertical merge operations for tree construction, which are very similar to the horizontal and
vertical concatenation operations from FSLPs. Whereas a top dag can be transformed in linear time into
an equivalent FSLP with a constant multiplicative blow-up (Proposition~\ref{lemma-top-dag-transform}),
the reverse transformation (from an FSLP to a top dag) needs time $O(\sigma \cdot n)$ and involves a multiplicative blow-up of size $O(\sigma)$
where $\sigma$ is the number of node labels of the tree  (Proposition~\ref{lemma:top-dag-transform-rev}). A simple example (Example~\ref{example-sigma}) shows that this
$\sigma$-factor is unavoidable.
The reason for the $\sigma$-factor is a technical restriction in the
definition of top dags: In contrast to FSLPs, top dags only allow
sharing of common subtrees but not of common subforests. Hence,
sharing between (large) subtrees which only differ in their root
labels may be impossible at all (as illustrated by
Example~\ref{example-sigma}), and this leads to the $\sigma$-blow-up
in comparison to FSLPs.
The impossibility of sharing subforests
would also complicate the technical details of our main algorithmic
results for FSLPs (in particular Proposition~\ref{lemma:fcns-transform-rev} and
Theorem~\ref{cor:main} which is discussed below) for which we make heavy use
of a particular normal form for FSLPs that exploits the sharing of proper subforests.
We therefore believe that at least for our purposes, FSLPs are a more adequate
formalism than top dags.

\subsection{\bf Testing equality modulo associativity and commutativity.} Our main algorithmic result for FSLPs can be formulated
as follows: Fix a set $\Sigma$ of node labels and
take a subset $\mathcal{C} \subseteq \Sigma$ of ``commutative'' node labels and a subset $\mathcal{A} \subseteq \Sigma$
of ``associative'' node labels. This means that for all $a \in \mathcal{A}$,  $c \in \mathcal{C}$ and all trees $t_1, t_2, \ldots, t_n$
({\it i}\/)~we do not distinguish between the trees $c\auf t_1, \ldots, t_n\zu$ and $c\auf t_{\sigma(1)}, \ldots, t_{\sigma(n)} \zu$, where
$\sigma$ is any permutation (commutativity), and
({\it ii}\/)~we do not distinguish the trees $a\auf t_1, \ldots, t_n\zu$ and $a\auf t_1, \ldots, t_{i-1}, a\auf t_i,\ldots, t_{j-1}\zu, t_j,
\ldots, t_n\zu$ for $1 \leq i \leq j \leq n+1$ (associativity). We then show that for two given FSLPs $F_1$ and $F_2$ that produce trees $t_1$ and $t_2$ (of possible
exponential size), one can check in polynomial time whether $t_1$ and $t_2$ are equal modulo commutativity and
associativity (Theorem~\ref{cor:main}).
Note that unordered tree isomorphism corresponds to the case $\mathcal{C} = \Sigma$ and $\mathcal{A}=\emptyset$
(in particular we generalize the result from~\cite{DBLP:conf/icalp/LohreyMP15} for ranked unordered trees).
Theorem~\ref{cor:main} also holds if the  trees $t_1$ and $t_2$ are given by top dags or TSLPs for the fcns-encodings,
since these formalisms can be transformed efficiently into FSLPs. Theorem~\ref{cor:main} also shows the utility of FSLPs even if one is only interested in say binary trees, which are represented by TSLPs. The law of associativity will yield very wide and flat trees that are no longer
compressible with TSLPs but are still compressible with FSLPs.

\section{Straight-line programs over algebras} \label{sec-SSLP}

We will produce strings, trees and forests by algebraic expressions over certain algebras. These expressions
will be compressed by directed acyclic graphs. In this section, we introduce the general framework, which will
be reused several times in this paper.

An algebraic structure is a tuple $\mathcal{A} = (A, f_1, \ldots, f_k)$ where $A$ is the universe and
every $f_i \colon A^{n_i} \to A$ is an operation of a certain arity $n_i$.
In this paper, the arity of all operations
will be at most  two. If $n_i = 0$, then $f_i$ is called a constant. Moreover, it will be convenient to
allow partial operations for the $f_i$.
Algebraic expressions over $\mathcal{A}$
are defined in the usual way: if $e_1, \ldots, e_{n_i}$ are algebraic expressions over $\mathcal{A}$, then
also $f_i(e_1, \ldots, e_{n_i})$ is an algebraic expressions over $\mathcal{A}$. For an algebraic
expression $e$, $\valX{e} \in A$  denotes the element to which $e$ evaluates (it can be undefined).

A {\em straight-line program} (SLP for short)
over  $\mathcal{A}$ is a tuple $P = (V, S, \rho)$, where $V$ is a set of {\em variables}, $S \in V$ is  the {\em start
  variable}, and $\rho$ maps every variable $A \in V$ to an expression of the form $f_i(A_1, \ldots, A_{n_i})$
(the so called {\em right-hand side} of $A$)
such that $A_1,\ldots,A_{n_i} \in V$ and the edge relation $E(P) =\{ (A, B) \in V \times V \mid B \text{ occurs in } \rho(A) \}$ is acyclic.
This allows to define for every variable $A \in V$ its value $\valXG{P}{A}$ inductively by
$\valXG{P}{A} = f_i(\valXG{P}{A_1}, \ldots, \valXG{P}{A_{n_i}})$
if $\rho(A) = f_i(A_1, \ldots, A_{n_i})$. Since the $f_i$ can be partially defined, the value of a variable
can be undefined. The SLP $P$ will be called {\em valid} if all values
$\valXG{P}{A}$ ($A \in V$) are defined. In our concrete setting, validity of an SLP
can be tested by a simple syntax check.
The value of $P$ is $\valX{P} = \valXG{P}{S}$.
Usually, we prove properties of SLPs
by induction along the partial order $E(P)^*$.

It will be convenient to allow for the right-hand sides $\rho(A)$ algebraic expressions over $\mathcal{A}$,
where the variables from $V$ can appear as atomic expressions. By introducing additional variables,
we can transform such an SLP into an equivalent SLP of the original form.
We define the size $|P|$ of an SLP $P$ as the total number of occurrences of operations $f_1,\ldots, f_k$ in all right-hand sides
(which is  the number of variables if all right-hand sides have the standard form $f_i(A_1, \ldots, A_{n_i})$).

Sometimes it is useful to view an SLP $P = (V,S,\rho)$ as a directed acyclic graph (dag)
$(V, E(P))$, together with the distinguished output node $S$, and the node labelling
that associates the label $f_i$ with the node $A \in V$ if $\rho(A) = f_i(A_1, \ldots, A_{n_i})$.
Note that the outgoing edges $(A,A_1), \ldots, (A,A_{n_i})$ have to be ordered since $f_i$ is in general not
commutative and that multi-edges have to be allowed. Such dags are also known as algebraic circuits
in the literature.

\subsection{\bf String straight-line programs.}
A widely studied type of SLPs are SLPs over a free monoid
$(\Sigma^*, \cdot, \varepsilon, (a)_{a \in \Sigma})$, where $\cdot$ is the concatenation operator (which, as usual, is not written
explicitly in expressions) and the empty string
$\varepsilon$ and every alphabet symbol $a \in \Sigma$ are added as constants.
We use the term {\em string straight-line programs}
(SSLPs for short) for these SLPs. If we want to emphasize the alphabet $\Sigma$, we speak of an SSLP
over $\Sigma$. In many papers, SSLPs are just called straight-line programs; see~\cite{lohrey_survey} for a survey.
Occasionally we consider SSLPs without a start variable $S$ and then write $(V,\rho)$.

\begin{example}
Consider the SSLP $G = (\{S,A,B,C\},S,\rho)$ over the alphabet $\{a,b\}$
with $\rho(S) = AAB$, $\rho(A) = CBB$, $\rho(B) = CaC$, $\rho(C) = b$.
We have $\valXG{G}{B} = bab$, $\valXG{G}{A} = bbabbab$, and
$\valX{G} = bbabbabbbabbabbab$. The size of $G$ is $8$ (six concatenation operators
are used in the right-hand sides, and there are two occurrences of constants).
\end{example}
In the next two sections, we introduce two types of algebras for trees and forests.

\section{Forest algebras and forest straight-line programs} \label{sec-FSLP}
\newcounter{fslpdef}
\renewcommand{\thefslpdef}{\roman{fslpdef}}
\newcommand{\nextstepdef}[1]{\refstepcounter{fslpdef}\thefslpdef\label{#1}}

\subsection{\bf Trees and forests.}
Let us fix a finite set $\Sigma$ of node labels for the rest of the paper.
We consider $\Sigma$-labelled rooted ordered trees, where ``ordered''
means that the children of a node are totally ordered.  Every node has a label from $\Sigma$.
Note that we make no rank assumption: the number of children of a node (also called its degree)
is not determined by its node label.
The set of nodes (resp. edges) of $t$ is denoted by $V(t)$ (resp., $E(t)$).
A {\em forest} is a (possibly empty) sequence of trees.
The size $|f|$ of a forest is the total number of nodes in $f$.
The set of all $\Sigma$-labelled forests is denoted by $\mathcal{F}_0(\Sigma)$
and the set of all $\Sigma$-labelled trees is denoted by $\mathcal{T}_0(\Sigma)$.
As usual, we can identify trees with expressions built up from symbols in $\Sigma$ and
parentheses.
Formally, $\mathcal{F}_0(\Sigma)$ and $\mathcal{T}_0(\Sigma)$
can be inductively defined as the following sets of strings over the alphabet $\Sigma \cup \{ \auf, \zu\}$.
\begin{itemize}
\item If $t_1, \ldots, t_n$ are $\Sigma$-labelled trees with $n \geq 0$, then the string
$t_1 t_2 \cdots t_n$ is a $\Sigma$-labelled forest (in particular, the empty string $\varepsilon$ is a $\Sigma$-labelled forest).
\item If $f$ is a $\Sigma$-labelled forest and $a \in \Sigma$, then $a\auf f\zu$ is a $\Sigma$-labelled tree
(where the singleton tree $a\auf\zu$ is usually written as $a$).
\end{itemize}
Let us fix a distinguished
symbol $x \not\in \Sigma$ for the rest of the paper (called the parameter).
The set of forests $f \in \mathcal{F}_0(\Sigma \cup \{x\})$ such that $x$ has a unique
occurrence in $f$ and this occurrence is at a leaf node is denoted by $\mathcal{F}_1(\Sigma)$.
Let $\mathcal{T}_1(\Sigma) = \mathcal{F}_1(\Sigma) \cap \mathcal{T}_0(\Sigma \cup \{x\})$.
Elements of $\mathcal{T}_1(\Sigma)$ (resp., $\mathcal{F}_1(\Sigma)$) are called
tree contexts (resp., forest contexts).
We finally define
$\mathcal{F}(\Sigma) = \mathcal{F}_0(\Sigma) \cup \mathcal{F}_1(\Sigma)$ and
$\mathcal{T}(\Sigma) = \mathcal{T}_0(\Sigma) \cup \mathcal{T}_1(\Sigma)$.
Following~\cite{DBLP:conf/birthday/BojanczykW08}, we define the {\em forest algebra}
$\mathsf{FA}(\Sigma) = ( \mathcal{F}(\Sigma), \conch, \concv, ( a )_{a \in \Sigma}, \varepsilon, x)$
as follows:
\begin{itemize}
\item $\conch$ is the horizontal concatenation operator: for forests $f_1, f_2 \in \mathcal{F}(\Sigma)$, $f_1 \conch f_2$
is defined if $f_1 \in \mathcal{F}_0(\Sigma)$ or $f_2 \in \mathcal{F}_0(\Sigma)$ and in this case we set
$f_1 \conch f_2 = f_1 f_2$ (i.e., we concatenate the corresponding sequences of trees).
\item $\concv$ is the vertical concatenation operator: for forests $f_1, f_2 \in \mathcal{F}(\Sigma)$, $f_1 \concv f_2$
is defined if $f_1 \in \mathcal{F}_1(\Sigma)$ and in this case $f_1 \concv f_2$ is obtained by replacing in $f_1$ the unique
occurrence of the parameter $x$ by the forest $f_2$.
\item Every $a \in \Sigma$ is identified with the unary function $a : \mathcal{F}(\Sigma) \to \mathcal{T}(\Sigma)$
that produces $a \auf f \zu$ when applied to $f \in \mathcal{F}(\Sigma)$.
\item $\varepsilon \in \mathcal{F}_0(\Sigma)$ and $x \in  \mathcal{F}_1(\Sigma)$ are constants of the forest algebra.
\end{itemize}
For better readability, we also write $f\langle g \rangle$ instead of $f \concv g$,
$fg$ instead of $f \conch g$,
and $a$ instead of $\addroot{a}{\varepsilon}$.
Note that a forest $f \in \mathcal{F}(\Sigma)$ can be also viewed as an algebraic expression over
$\mathsf{FA}(\Sigma)$, which evaluates to $f$ itself (analogously to the free term algebra).

\newcommand{\ffcns}{\text{fcns}}
\newcounter{fcnsdef}
\renewcommand{\thefcnsdef}{\roman{fcnsdef}}
\newcommand{\nextstepfcns}[1]{\refstepcounter{fcnsdef}\thefcnsdef\label{#1}}

\subsection{\bf First-child/next-sibling encoding.}
The first-child/next-sibling encoding transforms a forest
over some alphabet $\Sigma$ into a binary tree over $\Sigma \uplus \{\bot\}$.
We define
$\ffcns \colon \mathcal{F}_0(\Sigma) \to \mathcal{T}_0(\Sigma \uplus \{\bot\})$
inductively by:
({\it \nextstepfcns{item:fcns-base}}\/)~$\ffcns(\varepsilon) = \bot$ and
({\it \nextstepfcns{item:fcns-step}}\/)~$\ffcns(a\auf f \zu g)  = a\auf \ffcns(f) \ffcns(g) \zu$ for
$f, g \in \mathcal{F}_0(\Sigma)$, $a \in \Sigma$.
Thus, the left (resp., right) child of a node in $\ffcns(f)$ is the first child (resp., right sibling)
of the node in $f$ or a $\bot$-labelled leaf if it does not exist.
\begin{example}
If $f = a \auf bc\zu d\auf e \zu$ then
\begin{equation*}
\begin{split}
\ffcns(f) &= \ffcns(a \auf bc\zu d\auf e \zu)  =  a\auf \ffcns(bc) \ffcns(d\auf e \zu) \zu \\
&= a\auf b \auf \bot \ffcns(c) \zu  d \auf \ffcns(e)  \bot  \zu \zu =  a\auf b \auf \bot c\auf \bot\bot \zu \zu  d \auf e \auf \bot\bot\zu  \bot  \zu \zu .
\end{split}
\end{equation*}
\end{example}

\subsection{\bf Forest straight-line programs.}
A {\em forest straight-line program} over $\Sigma$, FSLP for short, is a valid straight-line program
over the algebra  $\mathsf{FA}(\Sigma)$
such that $\valX{F} \in \mathcal{F}_0(\Sigma)$.
Iterated vertical and horizontal concatenations allow to
generate forests, whose depth and width is exponential
in the FSLP size.
For an FSLP $F = (V,S,\rho)$ and $i \in \{0,1\}$
we define $V_i = \{ A \in V \mid \valXG{F}{A} \in \mathcal{F}_i(\Sigma) \}$.

\begin{example}
Consider the FSLP $F = (\{S,A_0,A_1, \ldots, A_n,B_0,B_1, \ldots, B_n\}, S, \rho)$ over $\{a,b,c\}$
with $\rho$ defined by
$\rho(A_0) = a$, $\rho(A_i) = A_{i-1} A_{i-1}$ for $1 \leq i \leq n$,
$\rho(B_0) = \addroot{b}{A_n x A_n}$, $\rho(B_i) = B_{i-1} \langle B_{i-1} \rangle$
for $1 \leq i \leq n$, and $\rho(S) = B_n \langle c \rangle$. We have
$\valX{F} = b \auf a^{2^n} b \auf a^{2^n} \cdots b \auf a^{2^n} c \, a^{2^n} \zu \cdots a^{2^n} \zu a^{2^n} \zu$,
where $b$ occurs $2^n$ many times.
\end{example}
\begin{example} \label{ex-FSLP}
  Consider the alphabet $\Sigma=\{a, b, c, d, e\}$. Let $n \geq 0$ be a natural number, and let
  $F = (V,S_1,\rhs)$ be the FSLP with
\begin{itemize}
\item $V_0=\{A_1, A_2, B, S_1\}$, $V_1=\{B_0, \ldots, B_n,\allowbreak C_0, \ldots, C_n\}$,
\item $\rhs(A_1) = \addroot{e}{\addroot{e}{a b} c}$,
\item $\rhs(A_2) = \addroot{e}{a \, \addroot{e}{b c}}$,
\item $\rhs(B_0) = A_1 x A_2$,
\item $\rhs(B_i) = B_{i-1} \langle B_{i-1} \rangle$ for $1\leq i\leq n$,
\item $\rhs(B) = B_{n} \langle A_1\rangle$,
\item $\rhs(C_0) = \addroot{d}{x B}$,
\item $\rhs(C_i) = C_{i-1}\langle C_{i-1} \rangle$ for $1\leq i\leq n$, and
\item $\rhs(S_1) = C_n \langle B \rangle$.
\end{itemize}
  Note that, although $F$ has size $O(n)$, $\valX{F}$ has exponential width {\em and} depth,
  as it is the tree
  \[
  \underbrace{d \auf d \auf \cdots d \auf d \auf}_{2^n \text{ many } d\auf } f \underbrace{f \zu f \zu \cdots f \zu f  \zu}_{2^n \text{ many } f\zu } ,
  \]
where $f = \valXG{F}{B}$ is the forest $(e \auf e \auf a b \zu c\zu)^{2^{n}+1}  (e \auf a \, e \auf b c\zu\zu)^{2^{n}}$.

  Now consider a second FSLP $F' = (V',S_2,\rhs')$ over $\Sigma$
  with
 \begin{itemize}
\item  $V_0'=\{D, E_0, \ldots, E_n, E, S_2\}$,
\item $V_1'=\{F_0, \ldots, F_n\}$,
\item    $\rhs(D) = \addroot{e}{a b c}$,
\item    $\rhs(E_0) = D D$,
\item    $\rhs(E_i) = E_{i-1} E_{i-1}$ for $1 \leq i\leq n$,
\item    $\rhs(E) = E_{n} D$,
\item    $\rhs(F_0) = \addroot{d}{E x}$,
\item    $\rhs(F_i) = F_{i-1} \langle F_{i-1} \rangle$ for $1\leq i\leq n$, and
\item    $\rhs(S_2) = F_n \langle E\rangle$.
\end{itemize}
Then $\valX{F'}$ is the tree
   \[ \underbrace{d \auf f' d \auf f' \cdots d \auf f' d \auf f'}_{2^n \text{ many } d\auf f'} f' \zu \zu \cdots \zu \zu ,
   \]
  where $f' = \valXG{F'}{E}$ is the forest $e \auf a b c\zu^{2^{n+1}+1}$.

  Note that if we consider $e$ as associative (meaning that $e\auf s \, e \auf t u \zu \zu = e \auf e \auf st \zu u \zu$ for all trees $s,t,u$), then $f$ and $f'$ represent the same forest. If in addition we consider
  $d$ as commutative (meaning that $d \auf st \zu = d\auf ts\zu$ for all trees $s,t$)
  then the FSLPs $F$ and $F'$ in fact represent the same unranked tree.
  Our main contribution is a polynomial time algorithm for performing this kind of equivalence check.
\end{example}

FSLPs generalize {\em tree straight-line programs} (TSLPs for short)
that have been used for the compression of ranked trees before, see e.g.~\cite{Lohrey15dlt}. We only need
TSLPs for binary trees. A TSLP over $\Sigma$ can then be defined as an FSLP
$T = (V,S,\rho)$ such that for every $A \in V$, $\rho(A)$ has the form
$a$, $\addroot{a}{BC}$, $\addroot{a}{xB}$,
$\addroot{a}{Bx}$, or $B \langle C \rangle$ with $a \in \Sigma$, $B,C \in V$.
TSLPs can be used in order to compress the fcns-encoding of an unranked tree;
see also \cite{DBLP:journals/is/LohreyMM13}. It is not hard to see that an FSLP $F$ that produces
a binary tree can be transformed into a TSLP $T$ such that $\valX{F} = \valX{T}$ and
$|T| \in O(|F|)$. This is an easy corollary of our normal form for FSLPs that we introduce next
(see also the proof of Proposition~\ref{lemma-fcns-transform}).

\subsection{\bf Factorization of SSLPs.}
\label{subsec:factorization}
Let $\Sigma$ be an alphabet, let $\Sigma_1 \subseteq \Sigma$ and $\Sigma_2 = \Sigma \setminus \Sigma_1$. Then every
string $w \in \Sigma^*$ has a unique factorization $w = v_0 a_1 v_1 \cdots a_n v_n$
with $n \geq 0$, $a_i \in \Sigma_1$ and $v_0, v_i \in \Sigma_2^*$ for $i \in \{1, \ldots, n\}$,
which we call the $\Sigma_1$-{\em factorization of} $w$.
Let $G = (V, \rhs)$ and $G' = (V', \rhs')$ be SSLPs over $\Sigma$. We call $G'$
a $\Sigma_1$-{\em factorization of} $G$ if $\valXG{G}{A} = \valXG{G'}{A}$
for all $A \in V$, and there are sets $\Up, \Low$ of ({\em upper} and
{\em lower}) variables such that $V' = V \uplus \Up \uplus \Low$ and
\begin{displaymath}
  \rhs'(V) \subseteq \Low \cup \Low \Sigma_1 \Low \cup \Low\,\Up \Sigma_1 \Low \qquad
  \rhs'(\Up) \subseteq \Sigma_1 \Low \cup \Up\,\Up \qquad
  \rhs'(\Low) \subseteq \{\varepsilon\} \cup \Sigma_2 \cup \Low \Low .
\end{displaymath}
Note that the partition $V' = V \uplus \Up \uplus \Low$ is uniquely determined
by $V'$ and $\rhs$. Moreover, $\valXG{G'}{A} \in \Sigma_2^*$ for every $A \in \Low$
and $\valXG{G'}{A} \in (\Sigma_1\Sigma_2^*)^*$ for every $A \in \Up$. This implies
that $G'$ describes the $\Sigma_1$-factorization $w = v_0 a_1 v_1 \cdots a_n v_n$
for every string $w = \valXG{G'}{A} = \valXG{G}{A}$ $(A \in V)$ in the
following sense:
If $\rhs'(A) = B \in \Low$, then $n = 0$ and $\valXG{G'}{B} = v_0$.
If $\rhs'(A) = BaC \in \Low \Sigma_1 \Low$, then $n = 1$, $\valXG{G'}{B} =
v_0$, $a = a_1$ and $\valXG{G'}{C} = v_1$. Finally, if $\rhs'(A) = BCaD
\in \Low\,\Up \Sigma_1 \Low$ then $n \geq 2$, $\valXG{G'}{B} = v_0$, $a = a_n$,
$\valXG{G'}{D} = v_n$ and there are variables $C_i,D_i$
with $\valXG{G'}{C} = \valXG{G'}{C_1} \cdots \valXG{G'}{C_{n-1}}$,
$\rhs(C_i) = a_i D_i$ and $\valXG{G'}{D_i} = v_i$ for $i \in \{1,\ldots,n-1\}$.

\begin{lemma}
\label{lemma:slp_split}
Given an SSLP $G = (V, \rhs)$ over $\Sigma$ and $\Sigma_1 \subseteq \Sigma$, one can compute in linear time a
$\Sigma_1$-factorization of $G$ of size $O(|G|)$.
\end{lemma}

\begin{proof}
Let $G = (V, \rhs)$ be an SSLP over $\Sigma$, $\Sigma_1 \subseteq \Sigma$ and $\Sigma_2
= \Sigma \setminus \Sigma_1$. W.l.o.g. we can assume that $\rhs(V) \subseteq VV \cup \Sigma$.
 For every string $w \in \Sigma^*$ with $\Sigma_1$-factorization $w = v_0 a_1 v_1 \cdots a_n v_n$
 let $\kl{w}, \km{w}, \kr{w} \in \Sigma^*$ and $\ks{w} \in \Sigma_1 \cup \{\varepsilon\}$ be defined
 as follows:
 \begin{itemize}
 \item If $n = 0$ then $\kl{w} = v_0$ and $\km{w} = \kr{w} = \ks{w} = \varepsilon$.
 \item If $n > 0$ then $\kl{w} = v_0$, $\km{w} = a_1 v_1 \cdots a_{n-1} v_{n-1}$,
                        $\ks{w} = a_n$ and $\kr{w} = v_n$.
 \end{itemize}
 Note that in both cases $w = \kl{w} \km{w} \ks{w} \kr{w}$ and
 $\kl{w}, \km{w}, \ks{w}, \kr{w}$ satisfy the following equations:
 \begin{itemize}
 \item If $w = \varepsilon$ then $\kl{w} = \km{w} = \ks{w} = \kr{w} = \varepsilon$.
 \item If $w = a \in \Sigma_1$ then $\ks{w} = a$ and $\kl{w} = \km{w} = \kr{w} = \varepsilon$.
 \item If $w = b \in \Sigma_2$ then $\kl{w} = b$ and $\km{w} = \ks{w} = \kr{w} = \varepsilon$.
 \item If $w = uv$ with $u, v \in \Sigma^*$ then
   \begin{itemize}
   \item if $\ks{u} = \varepsilon$ then also $\km{u} = \kr{u} = \varepsilon$, hence
     $\kl{w} = \kl{u}\kl{v}$, $\km{w} = \km{v}$, $\ks{w} = \ks{v}$ and
     $\kr{w} = \kr{v}$,
   \item if $\ks{u} \in \Sigma_1$ and $\ks{v} = \varepsilon$ then also $\km{v} = \kr{v} = \varepsilon$,
     hence $\kl{w} = \kl{u}$, $\km{w} = \km{u}$, $\ks{w} = \ks{u}$ and
     $\kr{w} = \kr{u} \kl{v}$,
   \item if $\ks{u}, \ks{v} \in \Sigma_1$ then $\kl{w} = \kl{u}$,
     $\km{w} = \km{u} \ks{u} \kr{u} \kl{v} \km{v}$, $\ks{w} = \ks{v}$
     and $\kr{w} = \kr{v}$.
   \end{itemize}
 \end{itemize}
 We use these equations as a guideline for the construction of
 the $\Sigma_1$-factorization $G' = (V \uplus \Up \uplus \Low,
\rhs')$ of $G$. Take new variables $\Kl{A}, \Km{A}, \Kr{A}, U_{BC}, L_{BC} \notin V$ and
let
\begin{equation*}
\begin{split}
 \Up & = \{\Km{A} \mid A \in V\} \cup \{U_{BC} \mid BC \in \rhs(V)\}, \\
\Low & = \{\Kl{A}, \Kr{A} \mid A \in V\} \cup \{L_{BC} \mid BC \in \rhs(V)\} .
\end{split}
\end{equation*}
 For every $A \in V$ we define $\Ks{A} \in \Sigma_1 \cup \{\varepsilon\}$ and the right-hand
 sides of the new variables as follows:
 \begin{itemize}
 \item If $\rhs(A) = \varepsilon$ then
   $\rhs'(\Kl{A}) = \rhs'(\Km{A}) = \Ks{A} = \rhs'(\Kr{A}) = \eps$.
 \item If $\rhs(A) = a \in \Sigma_1$ then $\Ks{A} = a$ and
   $\rhs'(\Kl{A}) = \rhs'(\Km{A}) = \rhs'(\Kr{A}) = \eps$.
 \item If $\rhs(A) = b \in \Sigma_2$ then $\rhs'(\Kl{A}) = b$ and
   $\rhs'(\Km{A}) = \Ks{A} = \rhs'(\Kr{A}) = \eps$.
 \item If $\rhs(A) = BC$ then
   \begin{itemize}
   \item if $\Ks{B} = \eps$ then $\rhs'(\Kl{A}) = \Kl{B} \Kl{C}$,
     $\rhs'(\Km{A}) = \Km{C}$, $\Ks{A} = \Ks{C}$
     and $\rhs'(\Kr{A}) = \Kr{C}$,
   \item if $\Ks{B} \in \Sigma_1$ and $\Ks{C} = \eps$ then
     $\rhs'(\Kl{A}) = \Kl{B}$, $\rhs'(\Km{A}) = \Km{B}$,
     $\Ks{A} = \Ks{B}$ and $\rhs'(\Kr{A}) = \Kr{B}\Kl{C}$,
   \item if $\Ks{B},\Ks{C} \in \Sigma_1$ then $\rhs'(\Kl{A}) = \Kl{B}$,
     $\rhs'(\Km{A}) = \Km{B} U_{BC} \Km{C}$ with $\rhs'(U_{BC}) =
     \Ks{B}L_{BC}$ and $\rhs'(L_{BC}) = \Kr{B}\Kl{C}$, $\Ks{A} =
     \Ks{C}$ and $\rhs'(\Kr{A}) = \Kr{C}$.
   \end{itemize}
 \end{itemize}
 Finally we define the new right-hand side for every $A \in V$: If
 $\Ks{A} = \eps$ then $\rhs'(A) = \Kl{A} \in \Low$. If $\Ks{A} \in \Sigma_1$ and
 $\valXG{G'}{\Km{A}} = \eps$ then $\rhs'(A) = \Kl{A}\Ks{A}\Kr{A} \in \Low \Sigma_1
 \Low$. Otherwise $\rhs'(A) = \Kl{A}\Km{A}\Ks{A}\Kr{A} \in \Low\,\Up \Sigma_1 \Low$.

 A straightforward induction on the structure of the SSLP $G$
 shows that
 $\valXG{G'}{\Kl{A}} = \kl{w}$,
 $\valXG{G'}{\Km{A}} = \km{w}$,
 $\Ks{A} = \ks{w}$ and
 $\valXG{G'}{\Kr{A}} = \kr{w}$ whenever
 $\valXG{G}{A} = w$. From this and the definition of the new
 right-hand sides $\rhs'(A)$ we finally obtain $\valXG{G'}{A}
 = (\valXG{G'}{\Kl{A}}) (\valXG{G'}{\Km{A}}) \Ks{A} (\valXG{G'}{\Kr{A}})
 = \kl{w}\km{w}\ks{w}\kr{w} = w$.
\end{proof}

\subsection{\bf Normal form FSLPs.}
In this subsection, we introduce a normal form for FSLPs that turns out to be crucial in the rest
of the paper. An FSLP
$F=(V, S, \rhs)$ is in {\em normal form} if $V_0 = V_0^{\top} \uplus V_0^{\bot}$ and
all right-hand sides have one of the following forms:
\begin{itemize}
\item $\rhs(A) = \varepsilon$, where $A \in V_0^\top$,
\item $\rhs(A) = BC$, where $A \in V_0^\top, B, C \in V_0$,
\item $\rhs(A) = B \langle C \rangle$, where $B \in V_1$ and either
$A, C \in V_0^\bot$ or $A,C \in V_1$,
\item $\rhs(A) = \addroot{a}{B}$, where $A \in V_0^\bot$, $a \in \Sigma$ and $B \in V_0$,
\item $\rhs(A) = \addroot{a}{B x C}$, where $A \in V_1$, $a \in \Sigma$ and $B, C \in V_0$.
\end{itemize}
Note that the partition $V_0 = V_0^{\top} \uplus V_0^{\bot}$ is uniquely determined
by $\rhs$. Also note that variables from $V_1$ produce tree contexts and variables
from $V_0^{\bot}$ produce trees, whereas variables from $V_0^{\top}$ produce
forests with arbitrarily many trees.

Let $F=(V,S,\rho)$ be a normal form FSLP. Every variable $A \in V_1$ produces a vertical
concatenation of (possibly exponentially many) variables, whose right-hand sides have the form
$\addroot{a}{B x C}$. This vertical concatenation is called the spine of $A$. Formally,
we split $V_1$ into $V_1^{\top} = \{ A \in V_1 \mid \exists B, C \in V_1 :
 \rhs(A) = B\langle C\rangle  \}$ and
$V_1^{\bot} = V_1 \setminus V_1^{\top}$. We then define
the {\em vertical SSLP}
$\spineslp{F} = (V_1^{\top}, \rhs_1)$ over $V_1^{\bot}$ with
$\rhs_1(A) = BC$ whenever $\rhs(A) = B\langle C\rangle$.
For every $A \in V_1$ the string $\valXG{\spineslp{F}}{A} \in ( V_1^\bot)^*$
is called the {\em spine} of $A$ (in $F$), denoted by
$\gspine{F}{A}$ or just $\fspine(A)$ if $F$ is clear from the
context. We also define the {\em horizontal SSLP}
$\ribslp{F} = (V_0^{\top}, \rhs_0)$ over $V_0^{\bot}$, where
$\rhs_0$ is the restriction of $\rhs$ to $V_0^{\top}$. For every
$A \in V_0$ we use $\ftrees(A)$ to denote the string
$\valXG{\ribslp{F}}{A} \in ( V_0^{\bot})^*$. Note that $\fspine(A) = A$
(resp., $\ftrees(A) = A$) for every $A \in V_1^{\bot}$ (resp., $A \in V_0^\bot$).

The intuition behind the normal form can be explained as follows: Consider a tree context $t \in \mathcal{T}_1(\Sigma) \setminus \{x\}$.
By decomposing $t$ along the nodes on the unique path from the root to the $x$-labelled leaf,
we can write $t$ as a vertical concatenation of tree contexts $a_1 \auf f_1 x g_1 \zu, \ldots,
a_n \auf f_n x g_n \zu$ for forests $f_1, g_1, \ldots, f_n, g_n$ and symbols $a_1, \ldots, a_n$.
In a normal form FSLP one would produce $t$ by first deriving a vertical concatenation $A_1 \lan \cdots \lan A_n \ran \cdots \ran$.
Every $A_i$ is then derived to $\addroot{a_i}{B_i x C_i}$, where $B_i$ (resp., $C_i$) produces the forest $f_i$ (resp., $g_i$).
Computing an FSLP for this decomposition for a tree context that is already given by an FSLP is the main step in the proof of the normal
form theorem below.  Another insight is that proper forest contexts from
$\mathcal{F}_1(\Sigma) \setminus \mathcal{T}_1(\Sigma)$
can be eliminated
without significant size blow-up.

\begin{theorem}
\label{lemma:fslp_normal_form}
From a given FSLP $F$ one can construct in linear time an FSLP $F'$ in
normal form such that $\valX{F'} = \valX{F}$ and $|F'| \in O(|F|)$.
\end{theorem}

\begin{proof}
To convert an FSLP to normal form,
we first introduce a {\em weak normal form},
where all right-hand sides have one of the following forms:
\begin{itemize}
\item $\rhs(A) = \varepsilon$, where $A \in V_0$,
\item $\rhs(A) = B\langle C\rangle$, where $A,C \in V_0$ and $B \in V_1$
\item $\rhs(A) = B\langle C\rangle$, where $A,B,C \in V_1$,
\item $\rhs(A) = \addroot{a}{x}$, where $A \in V_1$, $a \in \Sigma$,
\item $\rhs(A) = B x C$, where $A \in V_1$, $B, C \in V_0$.\label{item:siblings}
\end{itemize}
Converting an FSLP into weak normal form is straightforward:
By splitting up right-hand sides,  we can assume that
all right-hand sides have the form $\varepsilon, x, \addroot{a}{x}, B C$, or $B \langle C \rangle$
for $a \in \Sigma$, $B, C \in V$.  This transformation does not increase the size of the FSLP.
Right-hand sides of the form $\rho(A) = BC$, where w.l.o.g.~$B \in V_0$, can be replaced by
$\rho(A) = B' \langle C \rangle$ and $\rho(B') = B x$, where $B'$ is a new variable.

We may now assume that $F = (V,S,\rho)$ is in weak normal form.
Like we did with FSLPs in normal form, we
split $V_1$ into $V_1^{\top} = \{ A \in V_1 \mid \exists B, C \in V_1 :
 \rhs(A) = B\langle C\rangle  \}$ and
$V_1^{\bot} = V_1 \setminus V_1^{\top}$ and define
its spine SSLP as the SSLP
$\spineslp{F} = (V_1^{\top}, \rhs_1)$ over $V_1^{\bot}$ with
$\rhs_1(A) = BC$ whenever $\rhs(A) = B\langle C\rangle$.

Let $\Ver = \{ A \in V_1^{\bot} \mid \rhs(A) \text{ has the form } \addroot{a}{x} \}$
and $\Hor = V_1^{\bot} \setminus \Ver$. Thus, $\rho(A)$ has the form $BxC$ for $A \in \Hor$.
The idea of the construction is  to consider maximal factors of the form $A_0 A_1 \cdots A_n$
with $A_0 \in \Ver$ and $A_1, \ldots, A_n \in \Hor$ in $\valXG{\spineslp{F}}{A}$ (for some $A \in V_1^{\top}$).
In the FSLP $F$, such a factor corresponds to an iterated vertical concatenation
$A_0 \langle A_1 \langle \cdots \langle A_n \rangle \cdots \rangle \rangle$.
Assume that $\rho(A_0) = \addroot{a}{x}$ and $\rho(A_i) = B_i x C_i$ for $1 \leq i \leq n$.
Then, $A_0 \langle A_1 \langle \cdots \langle A_n \rangle \cdots \rangle \rangle$
can be rewritten into $\addroot{a}{B_1 B_2 \cdots B_n x C_n \cdots C_2 C_1}$.
We will introduce additional variables in order to produce the horizontal concatenations
$B_1 B_2 \cdots B_n$ and $C_n \cdots C_2 C_1$ and a variable with right-hand side $\addroot{a}{B x C}$.
Note that the latter form of right-hand sides is allowed in normal form FSLPs.

At this point, $\Ver$-factorizations turn out to be useful.
The maximal factors $A_0 A_1 \cdots A_n$ considered in Section~\ref{subsec:factorization} are explicitly generated
by the $\Ver$-factorization of the spine SSLP $\spineslp{F}$.
By Lemma~\ref{lemma:slp_split} we can
compute in linear time a $\Ver$-factorization
$G = (V_1^{\top} \uplus \Up \uplus \Low, \rhs_G)$ of $\spineslp{F}$ with $|G| \in O(|\spineslp{F}|) \leq O(|F|)$.
From $F$ and $G$ we obtain the FSLP
$F' = (V_0  \uplus \{A_\ell, A_r \mid A \in \Low\} \uplus \Up, S, \rhs')$
with new variables $A_\ell, A_r$ and  $\rhs'$
defined by:
\begin{enumerate}
  \item if $A \in \Low$ with $\rhs_G(A) = \varepsilon$
        then $\rhs'(A_\ell) = \rhs'(A_r) = \varepsilon$,
  \item \label{chain} if $A \in \Low$ with $\rhs_G(A) = B \in \Hor$ and $\rhs(B) = CxD$
        then $\rhs'(A_\ell) = C$ and $\rhs'(A_r) = D$,
  \item if $A \in \Low$ with $\rhs_G(A) = BC \in \Low \Low$
        then $\rhs'(A_\ell) = B_\ell C_\ell$ and $\rhs'(A_r) = C_rB_r$,
  \item if $A \in \Up$ with $\rhs_G(A) = BC \in \Ver \Low$ and $\rhs(B) = \addroot{a}{x}$
        then $\rhs'(A) = \addroot{a}{C_\ell xC_r}$,
  \item if $A \in \Up$ with $\rhs_G(A) = BC \in \Up \Up$
        then $\rhs'(A) = B \langle C \rangle$,
  \item if $A \in V_0$ with $\rhs(A) = \varepsilon$ then $\rhs'(A) = \eps$,
  \item if $A \in V_0$ with $\rhs(A) = B \langle A_0 \rangle$, $B \in V_1^{\bot}$ and $\rho(B) = \addroot{a}{x}$ then
     $\rhs'(A) = a \auf A_0 \zu$,
   \item \label{length3} if $A \in V_0$ with $\rhs(A) = B \langle A_0 \rangle$, $B \in V_1^{\bot}$ and $\rho(B) = C x D$ then
     $\rhs'(A) = C A_0 D$,
  \item \label{L} if $A \in V_0$ with $\rhs(A) = B \langle A_0 \rangle$, $B \in V_1^{\top}$ and  $\rhs_G(B) = C \in \Low$
          then $\rhs'(A) = C_\ell A_0C_r$,
  \item \label{LVL} if $A \in V_0$ with $\rhs(A) = B \langle A_0 \rangle$, $B \in V_1^{\top}$, $\rhs_G(B) = CDE \in \Low\,\Ver \Low$ and
  $\rhs(D) = \addroot{a}{x}$
          then $\rhs'(A) = C_\ell \addroot{a}{E_\ell A_0 E_r} C_r$,
  \item \label{LUVL} if $A \in V_0$ with $\rhs(A) = B \langle A_0 \rangle$, $B \in V_1^{\top}$, $\rhs_G(B) = CDD'E \in \Low\,\Up \Ver \Low$
             and $\rhs(D') = \addroot{a}{x}$
          then $\rho'(A) = C_\ell D \langle \addroot{a}{E_\ell A_0 E_r} \rangle C_r$.
 \end{enumerate}
 Note that this FSLP is not in normal form, but by further splitting up $\rho'(A)$ in points~\ref{length3}--\ref{LUVL} (and eliminating the ``chain
 definitions'' in point~\ref{chain}), we can
 obtain normal form.
 For instance, in point~\ref{LUVL}, we have to introduce new variables $A_1, \ldots, A_5$ and set
 $\rhs'(A) = A_1 C_r$,  $\rhs'(A_1) = C_\ell A_2$, $\rhs'(A_2) = D \langle A_3 \rangle$, $\rhs'(A_3) = a \langle A_4 \rangle$,
 $\rhs'(A_4) = A_5 E_r$, and $\rhs'(A_5) = E_\ell A_0$.
An easy induction on the partial order of the dag
shows that
\begin{itemize}
\item
      if $B \in \Low$ with $\valXG{\spineslp{F}}{B} = H_1 \cdots H_n \in \Hor^*$
      then $\valXG{F}{H_1 \langle \cdots \langle H_n \rangle \cdots \rangle} = \valXG{F'}{B_\ell xB_r}$,
\item
       if $A \in V_0$ then $\valXG{F}{A} = \valXG{F'}{A}$.
\end{itemize}
From the last point we finally obtain
$\valX{F} = \valXG{F}{S} = \valXG{F'}{S} = \valX{F'}$.
\end{proof}

\section{Cluster algebras and top dags} \label{sec-top-dag}

In this section we introduce top dags~\cite{BilleGLW15,Hubschle-Schneider15}
as an alternative grammar-based formalism for the compression of unranked trees.
A {\em cluster of rank $0$} is a tree $t \in \mathcal{T}_0(\Sigma)$ of size at least two.
A {\em cluster of rank $1$} is a tree $t \in \mathcal{T}_0(\Sigma)$ of size at least two together with
a distinguished leaf node that we call the {\em bottom boundary node} of $t$.
In both cases, the root of $t$ is called the {\em top boundary node} of $t$.
Note that in contrast to forest contexts there is no parameter $x$. Instead, one of the $\Sigma$-labelled
leaf nodes may be declared as the bottom boundary node. When writing a cluster of rank $1$ in term representation,
we underline the bottom boundary node. For instance $a \auf b \, c \auf \underline{a} \, b \zu \zu$
is a cluster of rank $1$. An {\em atomic cluster} is of the form $a \auf b \zu$ or $a \auf \underline{b} \zu$ for $a,b \in \Sigma$.
Let $\mathcal{C}_i(\Sigma)$ be the set of all clusters of rank $i \in \{0,1\}$
and let
$\mathcal{C}(\Sigma) = \mathcal{C}_0(\Sigma) \cup \mathcal{C}_1(\Sigma)$.
We write $\rank(s) = i$ if $s \in \mathcal{C}_i(\Sigma)$ for $i \in \{0,1\}$.
We define the {\em cluster algebra}
$\mathsf{CA}(\Sigma) = ( \mathcal{C}(\Sigma), \mergh, \mergv, (a \auf b \zu, a \auf \underline{b} \zu)_{a,b \in \Sigma})$
as follows:
\begin{itemize}
\item $\mergh$  is the horizontal merge operator: $s \mergh t$ is only defined if
$\rank(s)+\rank(t)\leq 1$ and $s,t$ are of the form
$s = a \auf f \zu$, $t = a \auf g \zu$, i.e., the root labels coincide. Then
$s \mergh t = a \auf f g \zu$. Note that at most one symbol in the forest $fg$ is underlined.
The rank of $s \mergh t$ is $\rank(s)+\rank(t)$. For instance,
$a \auf b \, c \auf \underline{a} \, b \zu \zu \mergh a \auf b \, c \zu = a \auf b \, c \auf \underline{a} \, b \zu b \, c \zu$.
\item $\mergv$ is the vertical merge operator: $s \mergv t$ is only defined if $s \in \mathcal{C}_1(\Sigma)$ and
the label of the root of $t$ (say $a$) is equal to the label of the bottom boundary node of $s$. We then
obtain $s \mergv t$ by replacing the unique occurrence of $\underline{a}$ in $s$ by $t$.
The rank of $s \mergv t$ is $\rank(t)$. For instance, $a \auf b \, c \auf \underline{a} \, b \zu \zu \mergv a \auf b \underline{c} \zu =
a \auf b \, c \auf a \auf b \underline{c} \zu  \, b \zu \zu$.
\item The atomic clusters $a \auf b \zu$ and $a \auf \underline{b} \zu$ are constants of the  cluster algebra.
\end{itemize}
A {\em top tree} for a tree $t \in \mathcal{T}_0$ is an algebraic expression $e$ over the algebra  $\mathsf{CA}(\Sigma)$
such that $\valX{e} = t$.
A {\em top dag} over $\Sigma$ is a straight-line program $D$ over the algebra  $\mathsf{CA}(\Sigma)$
such that $\valX{D} \in \mathcal{T}_0(\Sigma)$.
In our terminology, cluster straight-line program would be a more appropriate name,
but we prefer to call them top dags.

\begin{example}
Consider the top dag $D = (\{S,A_0, \ldots, A_n,B_0, \ldots, B_n\},S,\rho)$, where
$\rho(A_0) = b \auf a \zu$, $\rho(A_i) = A_{i-1} \mergh A_{i-1}$ for $1 \leq i \leq n$,
$\rho(B_0) = A_n \mergh b\auf\underline{b}\zu \mergh A_n$, $\rho(B_i) = B_{i-1} \mergv B_{i-1}$
for $1 \leq i \leq n$, and $\rho(S) = B_n \mergv b\auf c\zu$.
We have
$\valX{D} = b \auf a^{2^n} b \auf a^{2^n} \cdots b \auf a^{2^n} b \auf c \zu \, a^{2^n} \zu \cdots a^{2^n} \zu a^{2^n} \zu$,
where $b$ occurs $2^n+1$ many times.
\end{example}
\section{Relative succinctness} \label{sec-succinctness}

We have now three grammar-based formalisms for the compression of unranked trees:
FSLPs, top dags, and TSLPs for fcns-encodings. In this section we study their relative
succinctness. It turns out that up to multiplicative factors of size $|\Sigma|$ (number of node labels)
all three formalisms are equally succinct. Moreover, the transformations between the formalisms
can be computed very efficiently. This allows us to transfer algorithmic results for
FSLPs to top dags and TSLPs for fcns encodings, and vice versa.
We start with top dags:

\begin{proposition} \label{lemma-top-dag-transform}
For a given top dag $D$ one can compute in linear time an FSLP $F$ such that
 $\valX{F} = \valX{D}$ and $|F| \in O(|D|)$.
\end{proposition}
\begin{proof}
For $t \in \mathcal{T}(\Sigma)$ we denote  with $\fchop(t)$ the forest obtained
by removing from $t$ the root node.
Translating a cluster with a bottom boundary node to a tree with a parameter is done by
the function $\bx \colon \mathcal{C}_1(\Sigma) \to \mathcal{T}_1(\Sigma)$, where $\bx(t)$
replaces the bottom boundary node in $t$ labelled with $a \in \Sigma$ by the tree $a \auf x \zu$.
We translate a cluster to a forest by $\varphi \colon \mathcal{C}(\Sigma) \to \mathcal{F}(\Sigma)$,
where $\varphi(t) = \fchop(t)$ for $t \in \mathcal{C}_0(\Sigma)$ and $\varphi(t) = \fchop(\bx(t))$ for $t \in \mathcal{C}_1(\Sigma)$.
Then the following identities hold:
\begin{align}
\varphi(s) \conch \varphi(t) &= \varphi(s \mergh t) \label{phi1}\\
\varphi(s) \concv \varphi(t) &= \varphi(s \mergv t) \label{phi2}\\
\varphi( a \auf b \zu ) & = b \label{phi3}\\
\varphi( a \auf \underline{b} \zu ) & = b \auf x \zu \label{phi4}
\end{align}
Let $D=(V,S,\rhs)$ be a top dag and let $\alpha$ be the label of the root of $\valX{D}$, which can
be easily computed in linear time.
We define $F = (V \uplus \{ S' \}, S', \rhs')$, such that for every $A \in V$ we have
$\valXG{F}{A} = \varphi(\valXG{D}{A})$.
We set $\rhs'(S') = \addroot{\alpha}{S}$, which yields
\[
  \valX{F} = \valXG{F}{S'}
  = \valXG{F}{\addroot{\alpha}{S}}
  = \alpha \auf \valXG{F}{S} \zu
  = \alpha \auf \fchop(\valXG{D}{S}) \zu
  = \valXG{D}{S} = \valX{D}.
\]
We translate the right-hand sides of the top dag as follows:
\begin{itemize}
\item if $\rhs(A) = a \auf b \zu$ then $\rhs'(A) = b$,
\item if $\rhs(A) = a \auf \underline{b} \zu$ then $\rhs'(A) = \addroot{b}{x}$,
\item if $\rhs(A) = B \mergh C$ then $\rhs'(A) = B \conch C$,
\item if $\rhs(A) = B \mergv C$ then $\rhs'(A) = B \concv C$.
\end{itemize}
Then $\valXG{F}{A} = \varphi(\valXG{D}{A})$ for all $A \in V$ follows immediately from \eqref{phi1}--\eqref{phi4}.
\end{proof}

\begin{proposition}\label{lemma:top-dag-transform-rev}
For a given FSLP $F$ with $\valX{F} \in \mathcal{T}_0(\Sigma)$ and $|\valX{F}| \geq 2$
one can compute in time $O(|\Sigma| \cdot |F|)$ a top dag $D$ such that
$\valX{D} = \valX{F}$ and $|D| \in O(|\Sigma| \cdot |F|)$.
\end{proposition}
\begin{proof}
For every $a \in \Sigma$ we define the mapping $\psi_a \colon \mathcal{T}_1(\Sigma) \setminus \{x\}
\to \mathcal{C}_1(\Sigma)$ as follows:
for $t \in \mathcal{T}_1(\Sigma)$, $t \neq x$, let $\psi_a(t)$ be the rank-1 cluster obtained from
replacing in $t$ the label of the unique $x$-labelled node (which is not the root) by $a$ and declaring this node
as the bottom-boundary node.
Then, the following identities are obvious, where $s,t \in \mathcal{T}_1(\Sigma) \setminus \{x\}$, $u \in \mathcal{T}_0(\Sigma)$,
$|u| \geq 2$, and $b \in \Sigma$ is the label of the roots of $t$ and $u$:
\begin{align} \label{psi_a-1}
\psi_a( s \langle t \rangle ) &= \psi_b(s) \mergv \psi_a(t) \\
\label{psi_a-2}
s \langle u \rangle &= \psi_b(s) \mergv u
\end{align}
Moreover, for all forests $f,g \in \mathcal{F}_0(\Sigma)$ with $f \neq \varepsilon \neq g$ we have
\begin{equation}
\label{morph-a()}
a\auf  f g \zu = a \auf f \zu \mergh a \auf g \zu
\end{equation}
Let us now come to the construction for $T$.
By Theorem~\ref{lemma:fslp_normal_form} we can assume that the input FSLP $F = (V,S,\rho)$  is in normal form.
We can easily eliminate right-hand sides of the form $\varepsilon$ without a size increase.
This might lead to ``chain definitions'' of the form $\rho(A)=B$ which can be also eliminated without size increase.
After this preprocessing step, we may have also  right-hand sides of the form $\rho(A) = a \in \Sigma$ (with $A \in V_0^\bot$),
$\rho(A) = \addroot{a}{x}$, $\rho(A) = \addroot{a}{B x}$ (with $B \in V_0$), and $\rho(A) = \addroot{a}{x C}$ (with $C \in V_0$).
We still denote the resulting FSLP with $F$.
Since we started with an FSLP in normal form, we have $\valXG{F}{A} \in \mathcal{T}_0(\Sigma)$  for every $A \in V_0^\bot$ and
$\valXG{F}{A} \in \mathcal{T}_1(\Sigma) \setminus \{x\}$  for every $A \in V_1$. Hence,
for $A \in V_0^\bot \cup V_1$ we can define $\alpha_A \in \Sigma$ as the label of the root node in the tree (context) $\valXG{F}{A}$.
Also note that every forest $\valXG{F}{A}$ for $A \in V_0$ has size at least
one. Moreover, if $A \in V_0^\bot$ and $\rho(A) \not\in \Sigma$ then
the tree $\valXG{F}{A}$ has size at least two. Let $U_0^\bot = \{ A \in V_0^\bot \mid \rho(A) \not\in \Sigma \}$.

We define a top dag $D = (V', S, \rho')$, where $V' = V'_0 \cup V'_1$ with
\begin{align*}
V'_0 &=   U_0^\bot  \uplus \{ A^a \mid A \in V_0, a \in \Sigma \} \\
V'_1 &= \{ A_a \mid A \in V_1, a \in \Sigma \} .
\end{align*}
We will define the right-hand side mapping $\rho'$ of $D$ such that the following identities hold:
\begin{enumerate}[(i)]
\item $\valXG{D}{A} = \valXG{F}{A}$ for every $A \in U_0^\bot$,
\item $\valXG{D}{A^a} = a \auf \valXG{F}{A} \zu$ for every $A \in V_0$,
\item $\valXG{D}{A_a} = \psi_a(\valXG{F}{A})$ for every $A \in V_1$.
\end{enumerate}
In order to obtain these identities, we define  $\rho'$ as follows:
\begin{itemize}
\item if $\rho(A) = BC$ for $A,B,C \in V_0$
then $\rho'(A^a) = B^a \mergh C^a$,
\item if $A \in U_0^\bot$ then $\rho'(A^a) =  a \auf \underline{\alpha_A} \zu \mergv A$,
\item if $\rho(A) = b \in \Sigma$ then $\rho'(A^a) = a \auf b \zu$,
\item if $\rho(A) = \addroot{a}{B}$ (hence $A \in U_0^\bot$) then $\rho'(A) = B^a$,
\item if $\rho(A) = B \langle C \rangle$ for $A,C \in U_0^\bot$ and $B \in V_1$ then
 $\rho'(A) = B_{\alpha_C} \mergv C$,
 \item if $\rho(A) = B \langle C \rangle$, $\rho(C) = a \in \Sigma$ and $C \in V_1$ (hence $A \in U_0^\bot$)  then
 $\rho'(A) = B_a$,
 \item if $\rho(A) = B \langle C \rangle$ for $A,B,C \in V_1$ then
 $\rho'(A_a) = B_{\alpha_C} \langle C_a \rangle$,
\item  if $\rho(A) = \addroot{b}{B x C}$ for $A \in V_1$, $B,C \in V_0$ then
$\rho'(A_a) = B^b \mergh b \auf a \zu \mergh C^b$,
\item  if $\rho(A) = \addroot{b}{B x}$ for $A \in V_1$, $B \in V_0$ then
$\rho'(A_a) = B^b \mergh b \auf a \zu$,
\item  if $\rho(A) = \addroot{b}{x C}$ for $A \in V_1$, $C \in V_0$ then
$\rho'(A_a) = b \auf a \zu \mergh C^b$,
\item  if $\rho(A) = \addroot{b}{x}$ for $A \in V_1$ then
$\rho'(A_a) =b \auf a \zu$.
\end{itemize}
The correctness of this construction follows easily by induction, using \eqref{psi_a-1}--\eqref{morph-a()}.

To conclude the proof, note that since $\valX{F}$ is a tree of size two, the start symbol $S$ of $F$ must belong
to $U_0^\bot$. Hence, the above point ({\it i}\/) implies $\valX{D} = \valX{F}$.
\end{proof}
The following example shows that the size bound in Proposition~\ref{lemma:top-dag-transform-rev}
is sharp:

\begin{example} \label{example-sigma}
  Let $\Sigma = \{a, a_1,...,a_\sigma\}$ and let $t_{n} = a \auf a_1 \auf,
  a^m \zu \cdots a_\sigma \auf a^m \zu \zu$
  where $n \geq 1$  and $m = 2^n$.  For every $n > \sigma$ the tree
  $t_n$ can be produced by an FSLP of size $O(n)$: using $n = \log m$
  many variables we can produce the forest $a^{m}$ and then $O(n)$
  many additional variables suffice to produce $t_{n}$. On the other
  hand, every top dag for $t_{n}$ has size $\Omega(\sigma\cdot n)$: consider a top
  tree $e$ that evaluates to $t_{n}$. Then $e$ must contain a
  subexpression $e_i$ that evaluates to the subtree $a_i\auf a^m\zu$ ($1 \leq
  i \leq \sigma$) of $t_{n}$. The subexpression $e_i$ has to produce $a_i\auf
  a^m\zu$ using the $\mergh$-operation from copies of $a_i \auf a \zu$.
  Hence, the expression for $a_i \auf a^m\zu$ has size $n = \log_2 m$ and
  different $e_i$ contain no identical subexpressions. Therefore every
  top dag for $t_n$ has size at least $\sigma\cdot n$.
 \end{example}
In contrast, FSLPs and TSLPs for fcns-encodings turn out to be equally succinct up to constant factors:

\begin{proposition} \label{lemma-fcns-transform}
Let $f \in \mathcal{F}(\Sigma)$ be a forest and let $F$ be an FSLP (or TSLP) over $\Sigma \uplus \{\bot\}$
 with $\valX{F} = \ffcns(f)$. Then we can transform $F$ in linear time
 into an FSLP $F'$ over $\Sigma$ with $\valX{F'} = f$ and $|F'| \in O(|F|)$.
\end{proposition}
\begin{proof}
  Let $F = (V, S, \rhs)$ be an FSLP over $\Sigma \cup \{\bot\}$. By Theorem~\ref{lemma:fslp_normal_form},
  we may assume that $F$ is in
  normal form and every variable is reachable from
  $S$. This implies $|\ftrees(A)| \leq 2$ for every $A \in V_0$, because
  $\ffcns(f)$ is a binary tree. Hence we can compute the strings
  $\ftrees(A) = \valXG{\ribslp{F}}{A} \in ( V_0^\bot)^*$ with $A \in V_0^\top$
  all together in linear time, substitute $\ftrees(A)$ for each
  occurrence of $A$ in the right-hand sides, and finally erase the
  production for $A$. In particular, right-hand sides of the form $\varepsilon$ and
  $BC$ do not occur any more. Moreover, right-hand sides of the form $\addroot{a}{BxC}$ and $\addroot{a}{B}$
  will be transformed as follows by the above replacement: In the first case ($\addroot{a}{BxC}$)
  we have $a \in \Sigma$ and $|\ftrees(B)| + |\ftrees(C)| = 1$. Hence the substitution
  leads to $\addroot{a}{Dx}$ or $\addroot{a}{xD}$ with $D \in V_0^\bot$.
  In the second case ($\addroot{a}{B}$)
  either $a = \bot$ and $|\ftrees(B)| = 0$ or $a \in \Sigma$ and $|\ftrees(B)| =
  2$, hence the substitution leads to $\bot$ or $\addroot{a}{CD}$ with $C, D \in
  V_0^\bot$. Thus we finally obtain an FSLP in which all right-hand sides
  have one of the following forms:
\begin{itemize}
\item $\rhs(A) = \bot$
\item $\rhs(A) = \addroot{a}{B C}$
\item $\rhs(A) = \addroot{a}{B x}$
\item $\rhs(A) = \addroot{a}{x B}$
\item $\rhs(A) = B \langle C \rangle$
\end{itemize}
This is in fact a TSLP as defined in Section~\ref{sec-FSLP}.
We can now easily translate right-hand sides of
the above forms into right-hand sides of an FSLP $F'$ for $f$:
\begin{itemize}
\item $\rhs(A) = \bot$ becomes $\rhs(A) = \varepsilon$.
\item $\rhs(A) = \addroot{a}{B C}$ becomes $\rhs(A) = \addroot{a}{B} C$.
\item $\rhs(A) = \addroot{a}{B x}$ becomes $\rhs(A) = \addroot{a}{B} x$.
\item $\rhs(A) = \addroot{a}{x B}$ becomes $\rhs(A) = \addroot{a}{x} B$.
\item $\rhs(A) = B \langle C \rangle$ stays the same.
\end{itemize}
For the correctness of the construction, we have to show that $\ffcns(\valX{F'}) = \valX{F}$.
  In order to do this, we show the following properties:
  \begin{itemize}
  \item $\ffcns(\valXG{F'}{A}) = \valXG{F}{A}$ for all $A \in V_0$,
  \item  $\ffcns(\valXG{F'}{A}\langle f \rangle) = \valXG{F}{A}\langle \ffcns(f) \rangle$
    for all $A \in V_1$, $f \in \mathcal{F}_0(\Sigma)$.
  \end{itemize}
  These are shown using a simple induction and cases analysis:
  \begin{itemize}
    \item $\rhs(A) = \bot$: $\ffcns(\valXG{F'}{A}) = \ffcns(\varepsilon) = \bot = \valXG{F}{A}$.
    \item $\rhs(A) = \addroot{a}{B C}$: We obtain (``ind'' refers to induction on $B$ and $C$)
      \begin{equation*}
      \begin{split}
         \ffcns(\valXG{F'}{A}) & = \ffcns(\valXG{F'}{\addroot{a}{B} C}) \\
         & = \ffcns(a \auf \valXG{F'}{B} \zu \valXG{F'}{C}) \\
         & = a \auf \ffcns(\valXG{F'}{B}) \ffcns(\valXG{F'}{C}) \zu \\
         & \stackrel{\text{ind}}{=} a \auf \valXG{F}{B} \valXG{F}{C} \zu = \valXG{F}{A} .
\end{split}
\end{equation*}
    \item $\rhs(A) = \addroot{a}{B x}$: We obtain
     \begin{equation*}
     \begin{split}
      \ffcns(\valXG{F'}{A}\langle f \rangle) & = \ffcns(\valXG{F'}{\addroot{a}{B} x}\langle f \rangle) \\
      & = \ffcns(a \auf \valXG{F'}{B} \zu f) \\
      & = a \auf \ffcns(\valXG{F'}{B}) \ffcns(f) \zu \\
      & \stackrel{\text{ind}}{=} a \auf \valXG{F}{B} \ffcns(f) \zu \\
      & = \valXG{F}{\addroot{a}{B x}} \langle \ffcns(f)\rangle = \valXG{F}{A}\langle \ffcns(f)\rangle .
      \end{split}
      \end{equation*}
    \item $\rhs(A) = \addroot{a}{x B}$: We obtain
     \begin{equation*}
     \begin{split}
           \ffcns(\valXG{F'}{A}\langle f \rangle) & =  \ffcns(\valXG{F'}{\addroot{a}{x} B}\langle f \rangle) \\
             & = \ffcns(a \auf f \zu \valXG{F'}{B}) \\
             & = a \auf \ffcns(f) \ffcns(\valXG{F'}{B}) \zu  \\
             & \stackrel{\text{ind}}{=} a \auf \ffcns(f) \valXG{F}{B} \zu \\
             & = \valXG{F}{\addroot{a}{x B}}\langle \ffcns(f) \rangle = \valXG{F}{A} \langle \ffcns(f)\rangle .
      \end{split}
      \end{equation*}
      \item $\rhs(A) = B \langle C \rangle$ with $C \in V_0$: We obtain the following, where the first (resp., second) induction step
      uses induction on $B$ (resp., $C$):
       \begin{equation*}
       \begin{split}
        \ffcns(\valXG{F'}{A}) & = \ffcns(\valXG{F'}{B \langle C \rangle}) \\
       & = \ffcns(\valXG{F'}{B} \langle \valXG{F'}{C} \rangle ) \\
       & \stackrel{\text{ind}}{=} \valXG{F}{B} \langle \ffcns(\valXG{F'}{C}) \rangle \\
       & \stackrel{\text{ind}}{=} \valXG{F}{B} \langle \valXG{F}{C}\rangle \\
       & = \valXG{F}{B \langle C \rangle} = \valXG{F}{A}
       \end{split}
       \end{equation*}
    \item $\rhs(A) = B \langle C \rangle$ with $C \in V_1$: We obtain
     \begin{equation*}
     \begin{split}
      \ffcns(\valXG{F'}{A}\langle f \rangle) & = \ffcns(\valXG{F'}{B \langle C \rangle}\langle f \rangle) \\
         & = \ffcns( (\valXG{F'}{B}\langle \valXG{F'}{C} \rangle) \langle f \rangle) \\
         & = \ffcns(\valXG{F'}{B}\langle \valXG{F'}{C}\langle f \rangle \rangle) \\
         & \stackrel{\text{ind}}{=} \valXG{F}{B}\langle \ffcns(\valXG{F'}{C}\langle f\rangle ) \rangle \\
         & \stackrel{\text{ind}}{=} \valXG{F}{B} \langle \valXG{F}{C} \langle \ffcns(f)\rangle\rangle \\
        & =  \valXG{F}{B \langle C \rangle} \langle \ffcns(f)\rangle = \valXG{F}{A} \langle \ffcns(f)\rangle .
       \end{split}
       \end{equation*}
  \end{itemize}
  This concludes the proof of the proposition.
\end{proof}

\begin{proposition}\label{lemma:fcns-transform-rev}
    For every FSLP $F$ over $\Sigma$, we can construct in linear  time
    a TSLP $T$ over $\Sigma \cup \{\bot\}$ with
    $\valX{T} = \ffcns(\valX{F})$ and $|T| \in O(|F|)$.
  \end{proposition}
\begin{proof}
Let $F = (V,S,\rhs)$ be an FSLP over $\Sigma$. We may assume that $F$ is
already in normal form. We construct a TSLP $T = (V',S,\rhs')$ over
$\Sigma \cup \{\bot\}$ where
\begin{itemize}
\item $V_0' = \{ \kchop{A} \mid A \in V_0^\bot \} \uplus \{S\}$
\item $V_1' = \{ \kchop{A} \mid A \in V_1 \} \uplus \{ \kcont{A} \mid A \in V_0 \}$
\end{itemize}
with new variables $\kchop{A}, \kcont{A} \notin V$. For every $A \in V_1$
let $\ksibl{A} \in V_0$ be defined by
\begin{itemize}
\item $\ksibl{A} = C$  if $\rhs(A) = \addroot{a}{BxC}$, and
\item $\ksibl{A} = \ksibl{C}$  if $\rhs(A) = B \lan C \ran$ for $B,C \in V_1$.
\end{itemize}
Thus, $\valXG{F}{\ksibl{A}}$ is the list of right siblings of the
parameter $x$ in $\valXG{F}{A}$. For $A \in V_0^\bot$ we define the top
symbol $\alpha_A \in \Sigma$ as in the Proposition~\ref{lemma:top-dag-transform-rev}.
We then define $\rhs'$ by
\begin{itemize}
\item $\rhs'(S) = \kcont{S} \lan \bot \ran$ \smallskip
\item $\rhs'(\kchop{A}) = \kcont{B} \lan \bot \ran$
       \ if $\rhs(A) = \addroot{a}{B}$
      for $A \in V_0^\bot, B \in V_0$  \smallskip
\item $\rhs'(\kchop{A}) = \kchop{B} \lan \addroot{\rootof{C}}{\kchop{C}\, R_B^\pi\lan \bot \ran} \ran$
       \ if $\rhs(A) = B \lan C \ran$
       \smallskip
\item $\rhs'(\kchop{A}) = \kcont{B}$
       \ if $\rhs(A) = \addroot{a}{B x C}$
      for $A \in V_1^\bot$, $B,C \in V_0$  \smallskip
\item $\rhs'(\kcont{A}) = \addroot{\rootof{A}}{\kchop{A}\,x}$
       \ for every $A \in V_0^\bot$ \smallskip
\item $\rhs'(\kcont{A}) = x$
       \ if $\rhs(A) = \eps$
       for $A \in V_0^\top$  \smallskip
\item $\rhs'(\kcont{A}) = \kcont{B} \lan \kcont{C} \ran$
       \ if $\rhs(A) = B C$
       for $A \in V_0^\top$, $B,C \in V_0$.
\end{itemize}
Note that in $\rhs'(\kchop{A}) = \kchop{B} \lan \addroot{\rootof{C}}{\kchop{C}\, R_B^\pi\lan \bot \ran} \ran$ we
may have $\kchop{C} \in V'_0$ (if $C \in V_0^\bot$) or $\kchop{C} \in V'_1$ (if $C \in V_1$).
In the latter case we obtain for every $f \in \mathcal{F}_0(\Sigma)$,
\[
\valXG{F'}{\kchop{A}} \lan f \ran = \valXG{F'}{\kchop{B}} \lan \rootof{C} \auf \valXG{F'}{\kchop{C}}\lan f \ran\, \valXG{F'}{R_B^\pi}\lan \bot \ran \zu \ran .
\]
Let $\Delta \colon \mathcal{T}_0(\Sigma) \to \mathcal{F}_0(\Sigma)$ be defined by
$\Delta (a \auf f \zu) = f$. We will prove the following equations, which
express the role of the new variables in $V'$.
\begin{enumerate}[(1)]
\item\label{item:fcns-chop0}
      $\valXG{F'}{\kchop{A}} =
      \ffcns(\Delta (\valXG{F}{A}))$
      \ for every $A \in V_0^\bot$.
\item\label{item:fcns-cont}
      $\valXG{F'}{\kcont{A}} \lan \ffcns(f) \ran =
      \ffcns(\valXG{F}{A}\, f)$
      \ for every $A \in V_0, f \in \mathcal{F}_0(\Sigma)$.
\item\label{item:fcns-chop1}
      $\valXG{F'}{\kchop{A}} \lan \ffcns(t\, \valXG{F}{\ksibl{A}}) \ran =
      \ffcns(\Delta(\valXG{F}{A} \lan t \ran))$
      \ for every $A \in V_1, t \in \mathcal{T}_0(\Sigma)$.
\end{enumerate}
From \ref{item:fcns-cont} we obtain
$\valXG{F'}{\kcont{A}} \lan \bot \ran
 = \valXG{F'}{\kcont{A}} \lan \ffcns (\eps) \ran
 = \ffcns(\valXG{F}{A})$ for every $A \in V_0$.
This implies
$\valX{F'}
 = \valXG{F'}{S}
 = \valXG{F'}{\kcont{S}} \lan \bot \ran
 = \ffcns(\valXG{F}{S}) = \ffcns (\valX{F})$
which concludes the proof of
Proposition~\ref{lemma:fcns-transform-rev}. Hence only equations
\ref{item:fcns-chop0} to \ref{item:fcns-chop1} remain to be proved,
which is done by the following induction on the partial order induced by the dag $F$. Let
$A \in V$:

\medskip
\noindent
\ref{item:fcns-chop0} must be proved for every $A \in V_0^\bot$:
\begin{itemize}
\item If $\rhs(A) = \addroot{a}{B}$ then
  $\valXG{F'}{\kchop{A}}
   = \valXG{F'}{\kcont{B}}\lan \bot \ran
   = \ffcns(\valXG{F}{B})
   = \ffcns(\Delta (\valXG{F}{A}))$.\smallskip
\item If $\rhs(A) = B \lan C \ran$ with $B \in V_1$, $C \in V_0^\bot$ then
  \begin{align*}
    \valXG{F'}{\kchop{A}}
    = {} & \valXG{F'}{\kchop{B}}
          \lan \rootof{C} \auf \valXG{F'}{\kchop{C}}\,\valXG{F'}{R_B^\pi} \lan \bot \ran \zu \ran \\
    = {} & \valXG{F'}{\kchop{B}}
          \lan \rootof{C} \auf \ffcns (\Delta (\valXG{F}{C}))\,\ffcns(\valXG{F}{\ksibl{B}}) \zu \ran \\
         & \text{by induction for $C \in V_0^\bot$ and $\ksibl{B} \in V_0$} \\
    = {} & \valXG{F'}{\kchop{B}}
          \lan \ffcns (\rootof{C} \auf \Delta (\valXG{F}{C})\zu \, \valXG{F}{\ksibl{B}}) \ran \\
         & \text{by definition of $\ffcns$} \\
    = {} & \valXG{F'}{\kchop{B}} \lan \ffcns (\valXG{F}{C} \, \valXG{F}{\ksibl{B}}) \ran\\
    = {} & \ffcns (\Delta (\valXG{F}{B} \lan \valXG{F}{C} \ran )) \\
         & \text{by induction for $B \in V_1$}\\
    = {} & \ffcns (\Delta (\valXG{F}{A})).
  \end{align*}
\end{itemize}
\noindent
\ref{item:fcns-cont} must be proved for every $A \in V_0$:
\begin{itemize}
\item If $A \in V_0^\bot$ then
  \begin{align*}
    \valXG{F'}{\kcont{A}} \lan \ffcns (f) \ran
    = {} & \rootof{A} \auf \valXG{F'}{\kchop{A}} \, \ffcns(f) \zu \\
    = {} & \rootof{A} \auf \ffcns (\Delta (\valXG{F}{A})) \, \ffcns(f) \zu \\
         & \text{by equation \ref{item:fcns-chop0} for $\kchop{A}$} \\
    = {} & \ffcns (\rootof{A} \auf \Delta (\valXG{F}{A}) \zu \, f) \\
         & \text{by definition of $\ffcns$} \\
    = {} & \ffcns (\valXG{F}{A} \, f).
  \end{align*}
\item If $\rhs(A) = \eps$ then
  $\valXG{F'}{\kcont{A}} \lan \ffcns (f) \ran
  = \ffcns (f)
  = \ffcns (\eps\,f)
  = \ffcns (\valXG{F}{A}\,f)$.\smallskip
\item If $\rhs(A) = B C$ then
  \begin{align*}
    \valXG{F'}{\kcont{A}} \lan \ffcns (f) \ran
    = {} & \valXG{F'}{\kcont{B}} \lan \valXG{F'}{\kcont{C}} \lan \ffcns (f) \ran \ran \\
    = {} & \ffcns \lan \valXG{F}{B}\,\valXG{F}{C}\,f \ran \\
         & \text{by induction for $B$ and $C$} \\
    = {} & \ffcns (\valXG{F}{BC} \,f)\\
         & \text{by definition of $\ffcns$} \\
    = {} & \ffcns (\valXG{F}{A} \, f).
  \end{align*}
\end{itemize}

\noindent
\ref{item:fcns-chop1} must be proved for every $A \in V_1$:

\begin{itemize}
\item If $\rhs(A) = \addroot{a}{BxC}$ then
  \begin{align*}
    \valXG{F'}{\kchop{A}} \lan \ffcns (t\,\valXG{F}{\ksibl{A}}) \ran
    = {} & \valXG{F'}{\kcont{B}} \lan \ffcns (t\,\valXG{F}{\ksibl{A}}) \ran \\
    = {} & \ffcns (\valXG{F}{B}\,t\,\valXG{F}{C}) \\
         & \text{by induction for $B$ and because $\ksibl{A} = C$} \\
    = {} & \ffcns (\Delta (\valXG{F}{\addroot{a}{BxC}} \lan t \ran)) \\
    = {} & \ffcns (\Delta (\valXG{F}{A} \lan t \ran)).
  \end{align*}
\item If $\rhs(A) = B \lan C \ran$ with $A,B,C \in V_1$ then
  \begin{align*}
    \valXG{F'}{\kchop{A}} \lan \ffcns (t\,\valXG{F}{\ksibl{A}}) \ran
    = {} & \valXG{F'}{\kchop{B}}
           \lan \rootof{C} \auf \valXG{F'}{\kchop{C}} \lan \ffcns (t\,\valXG{F}{\ksibl{A}}) \ran\,
           \valXG{F'}{R_B^\pi} \lan \bot \ran \zu \ran \\
    = {} & \valXG{F'}{\kchop{B}}
           \lan \rootof{C} \auf \valXG{F'}{\kchop{C}} \lan \ffcns (t\,\valXG{F}{\ksibl{C}}) \ran\,
           \ffcns (\valXG{F}{\ksibl{B}}) \zu \ran \\
         & \text{by induction for $\ksibl{B}$ and because $\ksibl{A} = \ksibl{C}$}\\
    = {} & \valXG{F'}{\kchop{B}}
          \lan \rootof{C} \auf \ffcns (\Delta (\valXG{F}{C} \lan t \ran )) \,
          \ffcns (\valXG{F}{\ksibl{B}}) \zu \ran \\
         & \text{by induction for $C$}\\
    = {} & \valXG{F'}{\kchop{B}}
          \lan \ffcns (\rootof{C} \auf \Delta (\valXG{F}{C} \lan t \ran) \zu \, \valXG{F}{\ksibl{B}}) \ran  \\
         & \text{by definition of $\ffcns$} \\
    = {} & \valXG{F'}{\kchop{B}}
          \lan \ffcns (\valXG{F}{C} \lan t \ran  \, \valXG{F}{\ksibl{B}}) \ran  \\
    = {} & \ffcns (\Delta ( \valXG{F}{B} \lan \valXG{F}{C} \lan t \ran \ran )) \\
         & \text{by induction for $B$} \\
    = {} & \ffcns (\Delta ( \valXG{F}{A} \lan t \ran )).
  \end{align*}
\end{itemize}
This concludes the proof of the proposition.
\end{proof}
Proposition~\ref{lemma:fcns-transform-rev}  and
the construction from~\cite[Proposition~8.3.2]{tata07} allow to reduce
the evaluation of forest
automata on FSLPs (for a definition of forest and tree automata, see
\cite{tata07}) to the evaluation of ordinary tree automata on binary trees.
The latter problem can be solved in polynomial time ~\cite{LoMaSS12}, which yields:
\begin{corollary}
  \label{corollary:forest_automaton}
  Given a forest automaton $A$ and
  an FSLP (or top dag) $F$
  we can check in polynomial time whether $A$ accepts $\valX{F}$.
\end{corollary}
\begin{proof}
  First, we construct a TSLP $T$ for $\ffcns(\valX{F})$ using
  Proposition~\ref{lemma:fcns-transform-rev}.
  We also convert $A$ in polynomial time into a tree automaton $A'$
  such that $A'$ accepts $\ffcns(f)$ if and only if $A$ accepts $f$,
 using the construction from~\cite[Proposition~8.3.2]{tata07}. Finally, we use the result from
~\cite{LoMaSS12} to check in polynomial time whether $A'$ accepts
$\valX{T}$.
\end{proof}
In \cite{BilleGLW15}, a linear time algorithm is presented that constructs from a tree of size $n$ with $\sigma$ many node labels
a top dag of size $O(n / \log^{0.19}_\sigma n)$. In  \cite{Hubschle-Schneider15} this
bound was improved to $O(n \log\log n / \log_\sigma n)$ (for the same algorithm as in \cite{BilleGLW15}).
In \cite{LoReSi17} we recently presented an alternative construction that achieves the information-theoretic optimum of
$O(n / \log_\sigma n)$. Moreover, as in  \cite{BilleGLW15}, the constructed top dag satisfies the additional size bound
$O(d \cdot \log n)$, where $d$ is the size of the minimal dag of $t$.
With Proposition~\ref{lemma-top-dag-transform} and~\ref{lemma:fcns-transform-rev}
we get:

\begin{corollary}\label{corollary:optimal}
  Given a  tree $t$ of size $n$ with $\sigma$ many node labels, one can construct
  in linear time an FSLP for $t$ (or an TSLP for $\ffcns(t)$) of size
 $O(n / \log_\sigma n) \cap O(d \cdot \log n)$, where $d$ is the size of the minimal
 dag of $t$.
\end{corollary}

\section{Testing equality modulo associativity and commutativity}

In this section we will give an algorithmic application which
proves the utility of FSLPs (even if we deal with binary trees).
We fix two subsets
$\mathcal{A} \subseteq \Sigma$ (the set of {\em associative symbols}) and
$\mathcal{C} \subseteq \Sigma$ (the  set of {\em commutative symbols}).
This means that we impose the following identities
for all $a \in \mathcal{A}$, $c \in \mathcal{C}$,
all trees $t_1, \ldots, t_n \in \mathcal{T}_0(\Sigma)$, all permutations
$\sigma \colon \{1, \ldots, n\} \to \{1,\ldots,n\}$, and all $1 \leq i \leq j \leq n+1$:
\begin{align}
\label{eq-assoc}
a \auf t_1 \cdots t_n \zu & =
a \auf t_1 \cdots t_{i-1} a \auf t_i \cdots t_{j-1} \zu t_{j} \cdots t_n \zu   \\
\label{eq-comm}
c \auf t_1 \cdots t_n \zu & = c\auf t_{\sigma(1)} \cdots t_{\sigma(n)} \zu .
\end{align}
Note that the standard law of associativity for a binary symbol $\circ$ (i.e., $x \circ (y \circ z) = (x \circ y) \circ z$)
can be captured by making $\circ$ an (unranked) associative symbol in the sense of \eqref{eq-assoc}.
\subsection{Associative symbols}
\label{sec-assoc}

\newcounter{assocdef}
\renewcommand{\theassocdef}{\roman{assocdef}}
\newcommand{\nextstepassoc}[1]{\refstepcounter{assocdef}\theassocdef\label{#1}}

Below, we define the associative normal form $\nfa(f)$ of a forest $f$
and show that from an FSLP $F$ we can compute in linear time an FSLP $F'$
with $\valX{F'} = \nfa(\valX{F})$.
For trees $s, t \in \mathcal{T}_0(\Sigma)$ we have that $s = t$ modulo the
identities in \eqref{eq-assoc} if and only if $\nfa(s) = \nfa(t)$.
The generalization to forests is needed for the induction, where a
slight technical problem arises. Whether the forests $t_1 \cdots t_{i-1} a
\auf t_i \cdots t_{j-1} \zu t_{j} \cdots t_n$ and $t_1 \cdots t_n$ are equal modulo
the identities in \eqref{eq-assoc}
actually depends on the symbol on top of these two
forests. If it is an $a$, and $a\in\mathcal{A}$, then the two forests
are equal modulo associativity, otherwise not.
To cope with this problem, we use for every associative symbol $a \in \mathcal{A}$ a function
$\pull_a \colon \mathcal{F}_0(\Sigma) \to \mathcal{F}_0(\Sigma)$ that pulls up occurrences of $a$ whenever possible.

Let $\bullet \notin \Sigma$ be a new symbol.
For every $a \in \Sigma \cup \{ \bullet \}$ let $\pull_a \colon \mathcal{F}_0(\Sigma) \to
\mathcal{F}_0(\Sigma)$ be defined as follows, where $f \in \mathcal{F}_0(\Sigma)$ and $t_1, \ldots, t_n \in \mathcal{T}_0(\Sigma)$:
\[
\pull_a(b\auf f \zu)  = \begin{cases}
      \pull_a(f) & \text{if } a \in \mathcal{A} \text{ and } a = b, \\
      b \auf \pull_b(f) \zu & \text{otherwise,} \\
\end{cases} \qquad
\pull_a(t_1 \cdots t_n)  = \pull_a(t_1)  \cdots  \pull_a(t_n) .
\]
In particular, $\pull_a(\varepsilon) = \varepsilon$. Moreover, define
$\nfa \colon \mathcal{F}_0(\Sigma) \to \mathcal{F}_0(\Sigma)$ by $\nfa(f) = \pull_\bullet(f)$.

\begin{example}
Let $t = a \auf a \auf cd \zu b \auf cd \zu a \auf e \zu \zu$ and
$\mathcal{A} = \{a\}$. We obtain
\begin{align*}
 \pull_a(t) & = \pull_a(a \auf cd \zu b \auf cd \zu a \auf e \zu)  = \pull_a(a \auf cd \zu) \pull_a ( b \auf cd \zu ) \pull_a( a \auf e \zu ) \\
& = \pull_a(cd) b \auf  \pull_b(cd) \zu \pull_a(e)  = cd b \auf cd \zu e , \\
\pull_b(t) & = a \auf \pull_a(a \auf cd \zu b \auf cd \zu a \auf e \zu) \zu
  = a \auf cd b \auf cd \zu e \zu .
\end{align*}
\end{example}
To show the following simple lemma one considers the terminating and confluent rewriting system obtained
by directing the equations \eqref{eq-assoc} from right to left.
\begin{lemma} \label{lemma-only-assoc}
For two forests $f_1, f_2 \in \mathcal{F}_0(\Sigma)$,  $\nfa(f_1) = \nfa(f_2)$ if and only
if $f_1$ and $f_2$ are equal modulo the identities in \eqref{eq-assoc}
for all $a \in \mathcal{A}$.
\end{lemma}
\begin{proof}
Consider the (infinite) term rewriting system consisting of all rules
\begin{equation} \label{rewrite-sys}
a \auf t_1 \cdots t_{i-1} a \auf t_i \cdots t_{j-1} \zu t_{j} \cdots t_n \zu \to a \auf t_1 \cdots t_n \zu
\end{equation}
for $a \in \mathcal{A}$, $t_1, \ldots, t_n \in \mathcal{T}_0(\Sigma)$ and $1 \leq i \leq j \leq n+1$.
Let $\to$  be the resulting rewrite relation.
It is clearly terminating. Moreover, by considering all possible overlappings of left-hand
sides, one sees that the system is also confluent. Hence, every forest $f$ rewrites into a unique
normal form, which is in fact $\nfa(f)$. The lemma follows since $f_1$ and $f_2$ are equal modulo the identities in \eqref{eq-assoc}
if and only if they rewrite into the same normal forms, which means that $\nfa(f_1) = \nfa(f_2)$.
\end{proof}

\begin{lemma}
\label{lemma:assoc}
From a given FSLP $F = (V,S,\rho)$ over $\Sigma$
one can construct in
time $\mathcal{O}(|F|\cdot|\Sigma|)$
an FSLP $F'$ with $\valX{F'} = \nfa(\valX{F})$.
\end{lemma}
\begin{proof}
By Theorem~\ref{lemma:fslp_normal_form}, we assume that $F$ is in normal form.
We introduce new variables $A_a$ for all $a \in \Sigma \cup \{ \bullet \}$
and define the right-hand sides of $F'$ such that
$\valXG{F'}{A_a} = \pull_a(\valXG{F}{A})$
for all $A \in V_0$ and $\valXG{F'}{B_a \lan \pull_{b}(f) \ran} =
 \pull_a(\valXG{F}{B \lan f \ran})$ for all $B \in V_1$,
$f \in \mathcal{F}_0(\Sigma)$, where $b$ is the label of the parent node of the parameter
$x$ in $\valXG{F}{B}$. This parent node exists since $F$ is in normal form.
  For every $B \in V_1$ let $\omega_B$ be the symbol
  above $x$ in $\valXG{F}{B}$. These symbols exist by definition of
  the normal form, and they can be computed all together in
  linear time. Now let $F' = (V', S_\bullet, \rhs')$ where
  $V' = \{ A_a \mid A \in V, a \in \Sigma \cup \{ \bullet \} \}$,
  and $\rhs'$ is defined by
  \begin{itemize}
  \item $\rhs'(A_a) = \varepsilon$
    if $\rhs(A) = \varepsilon$,
  \item $\rhs'(A_a) = B_a C_a$
    if $\rhs(A) = BC$,
  \item $\rhs'(A_a) = B_a \langle C_{\omega_B} \rangle$
    if $\rhs(A) = B \langle C \rangle$,
  \item $\rhs'(A_a) = B_a$
    if $\rhs(A) = \addroot{a}{B}$ and $a \in \mathcal{A}$,
  \item $\rhs'(A_a) = \addroot{b}{B_b}$
    if $\rhs(A) = \addroot{b}{B}$
    with $b \neq a$ or $b \notin \mathcal{A}$,
  \item $\rhs'(A_a) = B_a x C_a$
    if $\rhs(A) = \addroot{a}{B x C}$
    with $a \in \mathcal{A}$,
  \item $\rhs'(A_a) = \addroot{b}{B_b x C_b}$
    if $\rhs(A) = \addroot{b}{B x C}$
    with $b \neq a$ or $b \not \in \mathcal{A}$.
  \end{itemize}
  An induction shows:
  \begin{enumerate}[(i)]
\item $\valXG{F'}{A_a} = \pull_a(\valXG{F}{A})$
  for all $A \in V_0$ and $a \in \Sigma \cup \{ \bullet \}$, and
  \item $\valXG{F'}{B_a \lan \pull_{\omega_B}(f) \ran} =
  \pull_a(\valXG{F}{B \lan f \ran})$ for all $B \in V_1$, $a \in \Sigma \cup \{ \bullet \}$ and
  $f \in \mathcal{F}_0(\Sigma)$.
\end{enumerate}
From ({\it i}\/) we obtain $\valX{F'} =
  \valXG{F'}{S_\bullet} = \pull_\bullet(\valXG{F}{S}) = \nfa(\valXG{F}{S}) = \nfa(\valX{F})$.
\end{proof}

\subsection{Commutative symbols} \label{sec-commutative}
To test whether two trees over $\Sigma$ are
equivalent with respect to commutativity, we define a {\em commutative normal form} $\fcanon(t)$ of a tree
$t \in \mathcal{T}_0(\Sigma)$  such that $\fcanon(t_1) = \fcanon(t_2)$
if and only if $t_1$ and $t_2$ are equivalent with respect to the
identities in \eqref{eq-comm}
for all  $c \in \mathcal{C}$.

We start with a general definition: Let $\Delta$ be a possibly infinite alphabet together
with a total order $<$. Let $\leq$ be the reflexive closure of $<$.
Define the function
$\ksort{<} \colon \Delta^* \to \Delta^*$ by $\ksort{<}(a_1\cdots a_n) = a_{i_1}\cdots a_{i_n}$ with
$\{i_1,\ldots,i_n\} = \{1,\ldots,n\}$ and $a_{i_1} \leq \cdots \leq a_{i_n}$.

\begin{lemma}\label{lemma:slp_sort}
  Let $G$ be an SSLP over $\Delta$ and let $<$ be some total
  order on $\Delta$. We can construct in time $\mathcal{O}(|\Delta|\cdot|G|)$
  an SSLP $G'$ such that $\valX{G'} = \ksort{<}(\valX{G})$.
\end{lemma}
\begin{proof}
Let $G = (V,S,\rho)$.
We define the SSLP $G' = (V', S', \rhs')$ over $\Delta$
where $V' = V \uplus \{A_a \mid A \in V, a \in \Delta\}$ with new variables $A_a \notin V$,
and $\rhs'$ defined by
\begin{itemize}
\item $\rhs'(A_a) = \eps$ if $\rhs(A) \in \{\eps\} \cup (\Delta \setminus \{a\})$,
\item $\rhs'(A_a) = a$ if $\rhs(A) = a$,
\item $\rhs'(A_a) = B_aC_a$ if $\rhs(A) = BC$,
\item $\rhs'(S') = A_{a_1} \ldots A_{a_n}$ if $\Delta = \{a_1, \ldots, a_n\}$ with $a_1 < \cdots < a_n$.
\end{itemize}
A straightforward induction shows that $\valXG{G'}{A_a} = a^{m_a}$ where
$m_a$ is the number of occurrences of $a$ in $\valXG{G}{A}$.
\end{proof}

In order to define the commutative normal form, we need a total order
on $\mathcal{F}_0(\Sigma)$. Recall that elements of $\mathcal{F}_0(\Sigma)$ are particular strings
over the alphabet $\Gamma := \Sigma \cup \{\auf, \zu\}$. Fix an arbitrary total
order on $\Gamma$ and let
$<_{\llex}$ be the {\em length-lexicographic order} on $\Gamma^*$ induced by
$<$: for $x,y \in \Gamma^*$ we have $x <_{\llex} y$
if $|x| < |y|$ or ($|x|=|y|$,  $x = u a v$, $y = u b v'$, and $a < b$ for $u,v,v' \in \Gamma^*$ and $a,b \in \Gamma$).
We now consider the restriction of $<_{\llex}$ to $\mathcal{F}_0(\Sigma) \subseteq \Gamma^*$.
For the proof of the following lemma one first constructs SSLPs for the strings
$\valX{F_1}, \valX{F_2} \in \Gamma^*$ (the construction is similar to the case of TSLPs, see
\cite{DBLP:journals/is/BusattoLM08}) and then uses \cite[Lemma~3]{DBLP:conf/icalp/LohreyMP15}
according to which SSLP-encoded strings can be compared in polynomial time with respect to
$<_{\llex}$.

\begin{lemma}
\label{corollary:fslp_equal}
For two FSLPs $F_1$ and $F_2$ we can check in polynomial time whether
$\valX{F_1} = \valX{F_2}$, $\valX{F_1} <_{\text{llex}} \valX{F_2}$ or
$\valX{F_2} <_{\text{llex}} \valX{F_1}$.
\end{lemma}
\begin{proof}
From $F_1$ and $F_2$ we first construct two SSLPs $G_1$ and $G_2$ that
produce $\valX{F_1}$ and $\valX{F_2}$, respectively, where the latter are viewed as a string over the alphabet
$\Sigma \cup \{ \auf, \zu \}$. The construction is similar to the case of TSLPs; see~\cite{DBLP:journals/is/BusattoLM08}:
Consider $F_1 = (V,S,\rho)$. By Theorem~\ref{lemma:fslp_normal_form} we can assume that $F_1$ is in normal form.
We define the SSLP $G_1 = (V', S, \rhs')$ over $\Sigma \cup \{ \auf, \zu \}$, where
$V' = V_0 \cup \{ A_1, A_2 \mid A \in V_1 \}$ and $\rhs'$ is defined as follows:
\begin{itemize}
\item If $\rhs(A) = \varepsilon$ or $\rhs(A) = BC$
  then  $\rhs'(A) = \rhs(A)$,
\item If $\rhs(A) = \addroot{a}{B}$ then $\rhs'(A) = a \auf B \zu$.
\item If $\rhs(A)= B \langle C\rangle$ with $C \in V_0$ then $\rhs'(A) = B_1 C B_2$.
\item If $\rhs(A) = \addroot{a}{BxC}$ then $\rhs'(A_1) = a \auf B$, and
  $\rhs'(A_2) = C \zu$.
\item If $\rhs(A) = B\langle C \rangle$ with $C \in V_1$ then $\rhs'(A_1) = B_1 C_1$
  and $\rhs'(A_2) = C_2 B_2$.
\end{itemize}
The correctness of the construction can be easily verified.

The rest of the proof follows immediately from \cite[Lemma~3]{DBLP:conf/icalp/LohreyMP15}:
Given SSLPs $G_1$ and $G_2$ over the same terminal alphabet $\Gamma$, we can
check in polynomial time whether $\valX{G_1} <_{\text{llex}} \valX{G_2}$, $\valX{G_2} <_{\text{llex}} \valX{G_1}$ or
$\valX{G_1} = \valX{G_2}$.
\end{proof}
From the restriction of $<_{\llex}$ to
$\mathcal{T}_0(\Sigma) \subseteq \Gamma^*$ we obtain the function $\ksort{<_{\llex}}$ on
$\mathcal{T}_0(\Sigma)^* = \mathcal{F}_0(\Sigma)$. We define
$\fcanon \colon \mathcal{F}_0(\Sigma) \to \mathcal{F}_0(\Sigma)$ by
\begin{align*}
  \fcanon(a\auf f \zu)  & =
  \begin{cases}
    a \auf \ksort{<_{\llex}}(\fcanon(f)) \zu & \text{if } a \in \mathcal{C} \\
    a \auf \fcanon(f) \zu         & \text{otherwise,}
  \end{cases} \\
 \fcanon(t_1\cdots t_n) & =  \fcanon(t_1) \cdots \fcanon(t_n).
\end{align*}
Obviously, $f_1, f_2 \in \mathcal{F}(\Sigma)$ are equal modulo the identities
in \eqref{eq-comm} for all $c \in \mathcal{C}$ if and only if
$\fcanon(f_1) = \fcanon(f_2)$. Using this fact and Lemma~\ref{lemma-only-assoc} it is not hard to show:
\begin{lemma}\label{lemma-assoc-comm}
For $f_1, f_2 \in \mathcal{F}_0(\Sigma)$ we have $\fcanon(\nfa(f_1)) =
\fcanon(\nfa(f_2))$ if and only if $f_1$ and $f_2$ are equal modulo
the identities in \eqref{eq-assoc}  and \eqref{eq-comm} for all $a \in \mathcal{A}$, $c \in \mathcal{C}$.
\end{lemma}
\begin{proof}
It suffices to show that $\fcanon(\nfa(f_1)) = \fcanon(\nfa(f_2))$ if $f_1$ and $f_2$ can be transformed
into each other by a single application of \eqref{eq-assoc} or \eqref{eq-comm}; let us write $f_1 =_{\text{\eqref{eq-assoc}}} f_2$
or $f_1 =_{\text{\eqref{eq-comm}}} f_2$, respectively, for the latter. The case $f_1 =_{\text{\eqref{eq-assoc}}} f_2$ is clear, since this implies
$\nfa(f_1)=\nfa(f_2)$ by Lemma~\ref{lemma-only-assoc}. Now assume that $f_1 =_{\text{\eqref{eq-comm}}} f_2$.
As in the proof of Lemma~\ref{lemma-only-assoc}, consider the infinite
rewriting system with the rules from  \eqref{rewrite-sys} and the associated rewrite relation $\to$.
The crucial observation is that $f_1 =_{\text{\eqref{eq-comm}}} f_2 \to f'_2$ implies that there exists $f'_1$
such that $f_1 \to f'_1 =_{\text{\eqref{eq-comm}}} f'_2$.
Since $f_2 \to^* \nfa(f_2)$, it follows that there exists $f'_1$ such that $f_1 \to^* f'_1 =_{\text{\eqref{eq-comm}}}  \nfa(f_2)$.
But this implies that $f'_1$ is irreducible with respect to $\to$, i.e., $f'_1 = \nfa(f_1)$.
We obtain $\nfa(f_1) =_{\text{\eqref{eq-comm}}}  \nfa(f_2)$ and hence $\fcanon(\nfa(f_1)) = \fcanon(\nfa(f_2))$.
\end{proof}
For our main technical result (Theorem~\ref{theorem:canonize}) we need a strengthening
of our FSLP normal form. Recall the notion of the {\em spine} from Section~\ref{sec-FSLP}.
We say that an FSLP $F = (V,S,\rhs)$ is in \emph{strong normal form} if
it is in normal form and for every $A \in V_0^\bot$ with $\rhs(A) = B \lan C \ran$
either $B \in V_1^\bot$ or $|\valXG{F}{C}| \geq |\valXG{F}{D}| - 1$ for every
$D \in V_1^\bot$ which occurs in $\fspine(B)$ (note that  $|\valXG{F}{D}| - 1$
is the number of nodes in $\valXG{F}{D}$ except for the parameter $x$).

\begin{lemma}
\label{lemma:simple_nf}
From a given FSLP $F = (V,S,\rhs)$ in normal form we can construct in
polynomial time an FSLP $F' = (V',S,\rhs')$ in strong normal form with
$\valX{F} = \valX{F'}$.
\end{lemma}
\begin{proof}
We modify the right-hand sides of variables $A \in V_0^\bot$ with $\rhs(A) = B \lan C \ran$
and $|\fspine(B)|  \geq 2$. Basically, we replace the vertical concatenations $B \lan C \ran$
by polynomially many vertical concatenations $B_i \lan C_i \ran$ which satisfy
the condition of the strong normal form.

$F'$ is obtained from $F$ by modifying (only) the right-hand sides of
variables $A \in V_0^\bot$ with $\rhs(A) = B \lan C \ran$ and $|\gspine{F}{B}| > 1$.
The modification for such a variable $A$ works as follows.

Let $\gspine{F}{B} = B_1 \cdots B_N$ ($N \geq 1$) and let $\{D_1,\ldots,D_m\} \sleq V_1^\bot$
($m \geq 1$) be the set of all variables which occur in
$\gspine{F}{B}$. For $1 \leq i \leq m$, let $p_i$ be the maximal position
$p \in \{1,\ldots,N\}$ such that $B_p = D_i$, i.e., the position of the last
occurence of $D_i$ in $\gspine{F}{B}$. The number $m$ and the
positions $p_i$ can be computed from $F$ in polynomial time, hence we
may assume that $p_m < \ldots < p_1$ by ordering the symbols $D_i$ in this
way. This means in particular that $p_1 = N$. Additionally, we set
$p_{m+1} = 0$.

For every $1 \leq i \leq m$ we can construct in polynomial time an SSLP
$G_i = (N_i,E_i,\rhs_i)$ over $V_1^\bot$ such that $\valX{G_i} = B_{p_{i+1}+1} \cdots B_{p_i-1}$
(see e.g. \cite[Lemma~1]{LoMaSS12}), hence
$\gspine{F}{B}
= \valX{G_m}B_{p_m} \cdots \valX{G_1}B_{p_1}
= \Mean{E_m}{G_m}D_m \cdots \Mean{E_1}{G_1}D_1$. We may assume that the
variable sets $N_i$ are pairwise disjoint and also disjoint from $V$,
and that  $\rhs_i(N_i) \sleq V_1^\bot \cup N_iN_i$ whenever $\valX{G_i} \neq \varepsilon$. Hence
we can add each $X \in N_i$ (with $\valX{G_i} \neq \varepsilon$) to the variable set
$V_1'$ of $F'$ and define its right-hand side by
\begin{itemize}
\item $\rhs'(X) = Y \lan Z \ran$ if $\rhs_i (X) = YZ$,
\item $\rhs'(X) = \rhs(D)$ if $\rhs_i (X) = D \in V_1^\bot$.
\end{itemize}
Thus we obtain $\Mean{B}{F'} = \Mean{E_m \lan D_m \lan \cdots E_1 \lan D_1 \ran \cdots \ran \ran}{F'}$.

Now we add new variables $A_i$ for $1 \leq i \leq m-1$ and $C_i$ for $1 \leq i
\leq m$ to the variable set $V_0'$ of $F'$ and define
\begin{itemize}
\item $\rhs'(C_1) = D_1 \lan C \ran$,
\item $\rhs'(C_i) = D_i \lan A_{i-1} \ran$ for $2 \leq i \leq m$,
\item $\rhs'(A_i) = E_i \lan C_i \ran$, if $\Mean{E_i}{G} \neq \varepsilon$, otherwise $\rhs'(A_i) = \rhs'(C_i)$\quad for $1 \leq i \leq m - 1$,
\item $\rhs'(A) = E_m \lan C_m \ran$, if $\Mean{E_m}{G} \neq \varepsilon$, otherwise $\rhs'(A) = \rhs'(C_m)$.
\end{itemize}
Obviously, $\Mean{C_i}{F'} = \Mean{D_i \lan \ldots D_1 \lan C \ran \ldots \ran}{F'}$ for
$1 \leq i \leq m$, which implies $|\Mean{C_i}{F'}| \geq |\Mean{D_j}{F'}| - 1$
for all $1 \leq j \leq i \leq m$ (equality holds if $i = m$ and $\Mean{C}{F'} =
\varepsilon$, since the parameter $x$ of $D_m$ disappears in this
case). Hence, the right-hand sides $\rhs'(A_i)$ and $\rhs'(A)$ meet the
definition of strong normal form. Moreover, $\Mean{A}{F'} =
\Mean{E_m \lan D_m \lan \ldots D_1 \lan C \ran \ldots \ran}{F'} = \Mean{B \lan C \ran}{F'}$. By
induction on the partial order of the dag, this implies $\Mean{A}{F'}
= \Mean{A}{F}$ for \emph{all} $A \in V$, because the right-hand
sides of other variables in $V$ are not modified. In particular,
$\valX{F'} = \Mean{S}{F'} = \Mean{S}{F} = \valX{F}$, which concludes
the proof.
\end{proof}

\begin{theorem}
\label{theorem:canonize}
From a given FSLP $F$ we can construct in polynomial time an
FSLP $F'$ with $\valX{F'} = \nfc(\valX{F})$.
\end{theorem}
\begin{proof}
Let $F = (V,S,\rhs)$. By Theorem~\ref{lemma:fslp_normal_form} and
Lemma~\ref{lemma:simple_nf} we may assume that $F$ is in strong normal
form. For every $A \in V_1$ let
\[ \args{A} = \{ t \in \mathcal{T}_0(\Sigma) \mid |t| \geq |\valXG{F}{D}| - 1
               \text{ for each symbol $D$ in $\fspine(A)\}$} \]
We want to construct an FSLP $F' = (V',S,\rhs')$ with $V_0 \sleq V_0'$ and $V_1 = V_1'$
such that
\begin{enumerate}[(1)]
\item\label{item:nfc0}
  $\valXG{F'}{A} = \nfc(\valXG{F}{A})$ \ for all $A \in V_0$,
  \item\label{item:nfc1}
  $\valXG{F'}{A} \lan \nfc(t) \ran = \nfc (\valXG{F}{A} \lan t \ran)$ \
  for all $A \in V_1$, $t \in \args{A}$.
\end{enumerate}
From~\ref{item:nfc0} we obtain
 $\valX{F'} = \valXG{F'}{S} =
\nfc(\valXG{F}{S}) = \nfc(\valX{F})$ which concludes the proof.

To define $\rhs'$, let $V^c = V_0^c \cup V_1^c$ with
$V_1^c = \{ A \in V_1 \mid \rhs(A) = \addroot{a}{BxC} \text{ with } a \in \Comm\}$ and
$V_0^c = \{ A \in V_0 \mid \rhs(A) = \addroot{a}{B} \text{ with } a \in \Comm \text{ or }
                      \rhs(A) = D \lan C \ran \text{ with } D \in V_1^c \}$ be the set of \emph{commutative variables}.
We set $\rhs'(A) = \rhs(A)$ for  $A \in V \setminus V^c$. For $A \in V^c$
we define $\rhs'(A)$ by induction along the partial order of the dag:
\begin{enumerate}
\item $\rhs(A) = \addroot{a}{B}$:
  Let $\Down{A}$ be the set of all $C \in V_0^\bot$ which are below $A$ in
  the dag, and let $w = \ftrees(B) =
  \valXG{\ribslp{F}}{B} \in \Down{A}^*$. By induction, $\rhs'$ is already
  defined on $\Down{A}$, and thus $\valXG{F'}{C}$ is defined for every
  $C \in \Down{A}$. By Lemma~\ref{corollary:fslp_equal}, we can compute
  in polynomial time a total order $<$ on $\Down{A}$ such that $C < D$
  implies $\valXG{F'}{C} \leq_\llex \valXG{F'}{D}$ for all $C,D
  \in \Down{A}$. By Lemma~\ref{lemma:slp_sort}, we can construct in
  linear time an SSLP $G_w = (V_w,S_w,\rhs_w)$ with $\valX{G_w} = \ksort{<}(w)$,
  and we may assume that all variables $D \in V_w$ are new. We add
  these variables to $V_0'$ together with their right hand
  sides $\rhs'(D) = \rhs_w(D)$, and we finally set $\rhs'(A) = \addroot{a}{S_w}$.
\item $\rhs(A) = B \lan C \ran$:
  Let $\rhs(B) = \addroot{a}{DxE}$. We define $G_w = (V_w,S_w,\rhs_w)$ as before,
  but with $w = \valXG{\ribslp{F}}{DCE}$ instead of $w =
  \valXG{\ribslp{F}}{B}$, and  we set $\rhs'(A) = \addroot{a}{S_w}$.
  \item $\rhs(A) = \addroot{a}{BxC}$:
  We define $G_w = (V_w,S_w,\rhs_w)$ as before, this time with $w =
  \valXG{\ribslp{F}}{BC}$, and we set $\rhs'(B) = \addroot{a}{S_w x}$.
\end{enumerate}
The main idea is that the strong normal form ensures that in right-hand sides
of the form $\addroot{a}{DxE}$ with $a \in \mathcal{C}$ one can move the parameter $x$
to the last position (see point 3 above),
since only trees that are larger than all trees produced from $D$ and $E$ are substituted for $x$.

Properties~\ref{item:nfc0} and \ref{item:nfc1} are proved by induction
along the partial order of the dag. We only consider the interesting cases, i.e.,
those in which $<_\llex$ plays a role.
\begin{enumerate}[(i)]
\item\label{item:proof-fcanon} $\rhs(A) = \addroot{a}{B}$ with $a \in \Comm$:\smallskip

Let $w = \Mean{B}{\ribslp{F}} = A_1 \cdots A_m$ with $m \geq 0$. Then
\begin{align*}
  \nfc (\Mean{A}{F})
  = {} & \nfc (a \auf \Mean{B}{F} \zu ) \\
  = {} & a \auf \ksort{<_\llex} (\nfc (\Mean{B}{F})) \zu
        \ \text{ by definition of $\nfc$ since $a \in \Comm$}\\
  = {} & a \auf \ksort{<_\llex} (\nfc(\Mean{A_1}{F}) \cdots \nfc(\Mean{A_m}{F})) \zu \\
  = {} & a \auf \ksort{<_\llex} (\Mean{A_1}{F'} \cdots \Mean{A_m}{F'}) \zu
        \ \text{ by induction for $A_1,\ldots,A_m$}\\
  = {} & a \auf \ksort{<_\llex} (\Mean{w}{F'}) \zu \\
  = {} & a \auf \Mean{\ksort{<}(w)}{F'} \zu \\
       & \text{since $A_i < A_j$ implies $\Mean{A_i}{F'} \leq_\llex \Mean{A_j}{F'}$ for $1 \leq i,j \leq m$} \\
  = {} & a \auf \Mean{S_w}{F'} \zu
        \ \text{by definition of $G_w = (V_w,S_w,\rhs_w)$}\\
  = {} & \Mean{A}{F'}
\end{align*}
\item $\rhs(A) = B \lan C \ran$ with $A,C \in V_0^\bot$ and $B \in V_1^c$,
  i.e., $\rhs(B) = \addroot{a}{DxE}$ with $a \in \Comm$:\smallskip

Let $w = \Mean{DCE}{\ribslp{F}} = A_1 \cdots A_m$ with $m \geq 0$. Then
\begin{align*}
  \nfc (\Mean{A}{F})
  & = \nfc (a \auf \Mean{DCE}{F} \zu ) \\
  & = a \auf \ksort{<_\llex} (\nfc (\Mean{DCE}{F})) \zu
       \ \text{ by definition of $\nfc$ since $a \in \Comm$} \\
  & = a \auf \ksort{<_\llex} (\nfc(\Mean{A_1}{F}) \cdots \nfc(\Mean{A_m}{F})) \zu \\
  & = a \auf \Mean{S_w}{F'} \zu
       \ \text{as in~\ref{item:proof-fcanon}} \\
  & = \Mean{A}{F'}
\end{align*}
\item $\rhs(A) = \addroot{a}{BxC}$ with $a \in \Comm$:\smallskip

Let $w = \Mean{BC}{\ribslp{F}} = A_1 \cdots A_m$ with $m \geq 0$, say
$\Mean{B}{\ribslp{F}} = A_1\cdots A_k$ and
$\Mean{C}{\ribslp{F}} = A_{k+1} \cdots A_m$ with $0 \leq k \leq m$.
For every $t \in \args{A}$ and $1 \leq i \leq m$ we have
$|\nfc(t)|
 = |t|
 \geq |\Mean{A}{F}| - 1
 > |\Mean{BC}{F}|
 \geq |\Mean{A_i}{F}|
 = |\nfc (\Mean{A_i}{F})|$,
hence $\nfc(\Mean{A_i}{F'}) \leq_\llex \nfc(t)$.
Thus we obtain
\begin{align*}
  \nfc ( \Mean{A}{F} \lan t \ran )
  = {} & \nfc ( a \auf \Mean{B}{F} \, t \, \Mean{C}{F} \zu ) \\
  = {} & a \auf \ksort{<_\llex}(\nfc (\Mean{B}{F} \, t \, \Mean{C}{F}) \zu
        \ \text{by definition of $\nfc$ since $a \in \Comm$} \\
  = {} & a \auf \ksort{<_\llex}(\nfc (\Mean{A_1}{F} \cdots \Mean{A_k}{F} \, t \,
                              \Mean{A_{k+1}}{F} \cdots \Mean{A_m}{F}) \zu \\
  = {} & a \auf \ksort{<_\llex}(\begin{array}[t]{l}
                        \nfc (\Mean{A_1}{F}) \cdots \nfc (\Mean{A_k}{F})\, \nfc (t) \\
                        \nfc (\Mean{A_{k+1}}{F}) \cdots \nfc (\Mean{A_m}{F}) )\zu
                       \end{array}  \\
       & \text{by definition of $\nfc$} \\
  = {} & a \auf \ksort{<_\llex}(\nfc (\Mean{A_1}{F}) \cdots \nfc (\Mean{A_m}{F})) \, \nfc(t)  \zu \\
       & \text{since $\nfc(\Mean{A_i}{F}) \leq_\llex \nfc(t)$ for $1 \leq i \leq m$} \\
  = {} & a \auf \Mean{\ksort{<}(w)}{F'}\,\nfc(t) \zu
        \ \text{as in~\ref{item:proof-fcanon}} \\
  = {} & \Mean{A}{F'} \lan \nfc (t) \ran
\end{align*}
\item $\rhs(A) = B \lan C \ran$ with $A,C \in V_0^\bot$ and $B \in V_1^\top$:\smallskip

Then  $\rhs'(A) = B \lan C \ran$ and $|\Mean{C}{F}| \geq |\Mean{D}{F}| - 1$ for
every $D$ which occurs in $\fspine(B)$, i.e., $\Mean{C}{F} \in \args{B}$.
Hence
\begin{align*}
  \nfc (\Mean{A}{F})
  & = \nfc (\Mean{B}{F} \lan \Mean{C}{F} \ran )
        \ \text{by induction for $C$}\\
  & = \Mean{B}{F'} \lan \nfc (\Mean{C}{F} ) \ran
        \ \text{by induction for $B$} \\
  & = \Mean{B}{F'} \lan \Mean{C}{F'} \ran \\
  & = \Mean{A}{F'}
\end{align*}
\item $\rhs(A) = B \lan C \ran$ with $A,B,C \in V_1$: \smallskip

Let $t \in \args{A} \sleq \args{B} \cap \args{C}$. Then $\Mean{C}{F} \lan t \ran \in
\args{B}$, and hence
\begin{align*}
  \nfc (\Mean{A}{F})
  = {} & \nfc (\Mean{B}{F} \lan \Mean{C}{F} \lan t \ran \ran ) \\
  = {} & \Mean{B}{F'} \lan \nfc(\Mean{C}{F} \lan t \ran ) \ran
        \ \text{ by induction for $B$} \\
  = {} & \Mean{B}{F'} \lan \Mean{C}{F'} \lan \nfc(t) \ran \ran
       \ \text{ by induction for $C$}\\
  = {} & \Mean{A}{F'} \lan \nfc (t) \ran
\end{align*}
\end{enumerate}
This concludes the proof of the theorem.
\end{proof}

\begin{theorem}\label{cor:main}
For trees $s,t$ we can test in polynomial time whether
$s$ and $t$ are equal modulo the identities in
\eqref{eq-assoc}  and
\eqref{eq-comm}, if $s$ and $t$
are given succinctly by one of the following three formalisms:
({\it i}\/) FSLPs, ({\it ii}\/) top dags, ({\it iii}\/)
TSLPs for the fcns-encodings of $s,t$.
\end{theorem}
\begin{proof}
By Proposition~\ref{lemma-top-dag-transform} and~\ref{lemma-fcns-transform}
it suffices to show Theorem~\ref{cor:main} for the case that $t_1$ and $t_2$ are
given by FSLPs $F_1$ and $F_2$, respectively. By  Lemma~\ref{lemma-assoc-comm}
and Lemma~\ref{corollary:fslp_equal} it suffices to compute in polynomial time
FSLPs $F'_1$ and $F'_2$ for $\fcanon(\nfa(t_1))$ and  $\fcanon(\nfa(t_2))$. This
can be achieved using Lemma~\ref{lemma:assoc} and
Theorem~\ref{theorem:canonize}.
\end{proof}

\section{Future work}

We have shown that simple algebraic manipulations (laws of associativity and commutativity)
can be carried out efficiently on grammar-compressed trees. In the future, we plan to investigate other
algebraic laws. We are optimistic that our approach can be extended by idempotent symbols (meaning that $a\auf f t t g\zu = a\auf f t g\zu$ for forests $f,g$ and a tree $t$).

Another interesting open problem concerns context unification modulo associative and commutative symbols.
The decidability of (plain) context-unification was a long standing open problem that was finally solved by Je\.{z}~\cite{Jez14},
who showed the existence of  a polynomial space algorithm. Je\.{z}'s algorithm uses his recompression technique
for TSLPs. One might try to extend this technique to FSLPs with the goal of proving decidability of context unification
for terms that also contain associative and commutative symbols.
For first-order unification and matching~\cite{GasconGS11},
context matching~\cite{GasconGS11}, and one-context unification~\cite{CreusGG12} there exist
algorithms for  TSLP-compressed trees that match the complexity of their
uncompressed counterparts.
One might also try to extend these results to
the associative and commutative setting.


\begin{thebibliography}{10}

\bibitem{DBLP:journals/mst/AbiteboulBV15}
S.~Abiteboul, P.~Bourhis, and V.~Vianu.
\newblock Highly expressive query languages for unordered data trees.
\newblock {\em Theor.~Comput. Syst.}, 57(4):927--966, 2015.

\bibitem{BilleGLW15}
P.~Bille, I.~L. G{\o}rtz, G.~M. Landau, and O.~Weimann.
\newblock Tree compression with top trees.
\newblock {\em Inf.~Comput.}, 243:166--177, 2015.

\bibitem{DBLP:conf/lata/BoiretHNT15}
A.~Boiret, V.~Hugot, J.~Niehren, and R.~Treinen.
\newblock Logics for unordered trees with data constraints on siblings.
\newblock In {\em Proc.~LATA 2015}, LNCS  8977, 175--187. Springer, 2015.

\bibitem{DBLP:conf/birthday/BojanczykW08}
M.~Boja{\'{n}}czyk and I.~Walukiewicz.
\newblock Forest algebras.
\newblock In {\em Proc.~Logic and Automata: History and Perspectives
  [in Honor of Wolfgang Thomas].}, volume~2 of {\em Texts in Logic and Games},
  107--132. Amsterdam University Press, 2008.

\bibitem{DBLP:journals/mst/BonevaCS15}
I.~Boneva, R.~Ciucanu, and S.~Staworko.
\newblock Schemas for unordered {XML} on a {DIME}.
\newblock {\em Theor.~Comput. Syst.}, 57(2):337--376, 2015.

\bibitem{DBLP:journals/is/BusattoLM08}
G.~Busatto, M.~Lohrey, and S.~Maneth.
\newblock Efficient memory representation of {XML} document trees.
\newblock {\em Information Systems}, 33(4-5):456--474, 2008.

\bibitem{tata07}
H.~Comon, M.~Dauchet, R.~Gilleron, F.~Jacquemard, C.~L{\"o}ding, D.~Lugiez,
  S.~Tison, and M.~Tommasi.
\newblock Tree automata techniques and applications.
\newblock Available at: http://www.grappa.univ-lille3.fr/tata, 2007.

\bibitem{CreusGG12}
C.~Creus, A.~Gasc{\'{o}}n, and G.~Godoy.
\newblock One-context unification with {STG}-compressed terms is in {NP}.
\newblock In {\em Proc.~RTA 2012}, LIPIcs~15, 149--164. Schloss Dagstuhl -- Leibniz-Zentrum f\"ur Informatik, 2012.

\bibitem{GasconGS11}
A.~Gasc{\'{o}}n, G.~Godoy, and M.~Schmidt{-}Schau{\ss}.
\newblock Unification and matching on compressed terms.
\newblock {\em {ACM} Transactions on Computational Logic}, 12(4):26:1--26:37,
  2011.

\bibitem{Hubschle-Schneider15}
L.~H{\"{u}}bschle{-}Schneider and R.~Raman.
\newblock Tree compression with top trees revisited.
\newblock In {\em Proc.~SEA 2015}, LNCS 9125, 15--27. Springer, 2015.

\bibitem{Jez14}
A.~Je{\.{z}}.
\newblock Context unification is in {PSPACE}.
\newblock In {\em Proc.~ICALP 2014, Part {II}}, LNCS 8573, 244--255. Springer, 2014.

\bibitem{lohrey_survey}
M.~Lohrey.
\newblock Algorithmics on {SLP}-compressed strings: a survey.
\newblock {\em Groups Complexity Cryptology}, 4(2):241--299, 2012.

\bibitem{Lohrey15dlt}
M.~Lohrey.
\newblock Grammar-based tree compression.
\newblock In {\em Proc.~DLT 2015}, LNCS 9168, 46--57. Springer, 2015.

\bibitem{DBLP:journals/is/LohreyMM13}
M.~Lohrey, S.~Maneth, and R.~Mennicke.
\newblock {XML} tree structure compression using RePair.
\newblock {\em Information Systems}, 38(8):1150--1167, 2013.

\bibitem{DBLP:conf/icalp/LohreyMP15}
M.~Lohrey, S.~Maneth, and F.~Peternek.
\newblock Compressed tree canonization.
\newblock In {\em Proc.~ICALP 2015, Part {II}},  337--349. Springer, 2015.

\bibitem{DBLP:conf/icdt/LohreyMR17}
M.~Lohrey, S.~Maneth, and C.~P. Reh.
\newblock Compression of unordered {XML} trees.
\newblock In {\em Proc.~ICDT 2017}, LIPIcs 68, 18:1--18:17.
Schloss Dagstuhl -- Leibniz-Zentrum f\"ur Informatik, 2017.

\bibitem{LoMaSS12}
M.~Lohrey, S.~Maneth, and M.~Schmidt-Schau{\ss}.
\newblock Parameter reduction and automata evaluation for grammar-compressed
  trees.
\newblock {\em J.~Comput.~Syst.~Sci.}, 78(5):1651--1669,
  2012.

\bibitem{LoReSi17}
M.~Lohrey, P.~Reh, and K.~Sieber.
\newblock Optimal top dag construction.
\newblock \url{https://arxiv.org/abs/1712.05822}, arXiv.org, 2017.

\bibitem{DBLP:journals/dke/SundaramM12}
S.~Sundaram and S.~K. Madria.
\newblock A change detection system for unordered {XML} data using a relational
  model.
\newblock {\em Data \& Knowledge Engineering}, 72:257--284, 2012.

\bibitem{DBLP:journals/tcyb/ZhangDW15}
S.~Zhang, Z.~Du, and J.~T. Wang.
\newblock New techniques for mining frequent patterns in unordered trees.
\newblock {\em IEEE Trans.~Cybern.}, 45(6):1113--1125, 2015.

\end{thebibliography}
\end{document}